\documentclass[aps,prx,superscriptaddress,nofootinbib,twocolumn]{revtex4-2}

\usepackage{natbib}
\setcitestyle{numbers}

\usepackage[utf8]{inputenc}
\usepackage{amsmath,amssymb,amsthm, mathtools, physics}
\usepackage{dsfont}
\usepackage{braket}
\usepackage{bm}
\usepackage{graphicx, color, xcolor}
\usepackage{comment}
\usepackage[caption=false]{subfig}
\usepackage{hyperref}
\hypersetup{
    colorlinks=true,
    citecolor=blue,
    linkcolor=magenta,
    urlcolor=magenta,
}

\newtheorem{lemma}{Lemma}
\newtheorem{theorem}{Theorem}
\newtheorem{definition}{Definition}
\newtheorem{corollary}{Corollary}

\makeatletter
\newsavebox{\@brx}
\newcommand{\llangle}[1][]{\savebox{\@brx}{$\m@th{#1\langle}$}%
  \mathopen{\copy\@brx\kern-0.5\wd\@brx\usebox{\@brx}}}
\newcommand{\rrangle}[1][]{\savebox{\@brx}{$\m@th{#1\rangle}$}%
  \mathclose{\copy\@brx\kern-0.5\wd\@brx\usebox{\@brx}}}
\makeatother

\newcommand{\T}{\mathrm{T}}

\begin{document}

\title{Learning functions of Hamiltonians with Hamiltonian Fourier features}

\author{Yuto Morohoshi}
\email{u542400c@ecs.osaka-u.ac.jp}
\affiliation{Graduate School of Engineering Science, The University of Osaka, 1-3 Machikaneyama, Toyonaka, Osaka 560-8531, Japan}

\author{Akimoto Nakayama}
\email{akimoto.nakayama@gmail.com}
\affiliation{Graduate School of Engineering Science, The University of Osaka, 1-3 Machikaneyama, Toyonaka, Osaka 560-8531, Japan}
\affiliation{Center for Quantum Information and Quantum Biology, The University of Osaka, Japan}

\author{Hidetaka Manabe}
\affiliation{Graduate School of Engineering Science, The University of Osaka, 1-3 Machikaneyama, Toyonaka, Osaka 560-8531, Japan}

\author{Kosuke Mitarai}
\email{mitarai.kosuke.es@osaka-u.ac.jp}
\affiliation{Graduate School of Engineering Science, The University of Osaka, 1-3 Machikaneyama, Toyonaka, Osaka 560-8531, Japan}
\affiliation{Center for Quantum Information and Quantum Biology, The University of Osaka, Japan}

\date{\today}

\begin{abstract}
We propose a quantum machine learning task that is provably easy for quantum computers and arguably hard for classical ones.  
The task involves predicting quantities of the form $\mathrm{Tr}[f(H)\rho]$, where $f$ is an unknown function, given descriptions of $H$ and $\rho$.
Using a Fourier-based feature map of Hamiltonians and linear regression, we theoretically establish the learnability of the task and implement it on a superconducting device using up to 40 qubits.  
This work provides a machine learning task with practical relevance, provable quantum easiness, and near-term feasibility.  
\end{abstract}

\maketitle

\section{Introduction}
Machine learning (ML) has become an indispensable tool in various fields.
Quantum machine learning is a field that aims to enhance machine learning techniques by utilizing quantum computers. 
While the field has found many algorithms with potential quantum speedups \cite{Cerezo2022-xy,Biamonte2017-ht,Cerezo2021-gk}, what type of practical ML tasks are provably advantageous to perform on quantum computers remains an open question.

Nonetheless, several ML tasks that separate classical and quantum capabilities have been found, although they may not be practically relevant.
For example, Servedio and Gortler \cite{Servedio2004Equivalences} have shown that there is an ML task that is provably hard to solve on classical computers but can be performed efficiently on quantum computers, relying on the result of \cite{Kearns1994} which has shown one can construct a classically-hard ML task based on integer factoring.
More recently, Liu \textit{et al.} \cite{Liu2021-pe} have proposed another task based on discrete logarithm and shown rigorously that the idea of quantum feature \cite{Schuld2019Quantum,Havlicek2019Supervised,Cerezo2021-gk} can solve it efficiently.
The following works \cite{YamasakiIsogaiMurao2023,Gyurik2023Exponential} have provided a more general framework for constructing such tasks from general computational problems that are efficiently solvable by quantum computers, although they still lack very practical implications.
A learning task with possible practicality has been proposed by Molteni \textit{et al.} \cite{MolteniGyurikDunjko2024}.
They have suggested a learning task that asks us to predict expectation values of an unknown observable for known quantum states and shown that it is a classically hard but quantumly easy task.

In this work, we propose another ML task with potential practical implications that is provably easy to perform on quantum computers.
The proposed ML task is intuitively described as follows: we know the Hamiltonian $H$ of a quantum system under experimental investigation, but we do not know the relationship between the Hamiltonian and the experimental outputs.
More specifically, it asks us to predict $y=\mathrm{Tr}[f(H)\rho]$ for some unknown function $f:\mathbb{R}\to\mathbb{R}$, while both $H$ and $\rho$ are drawn from a probability distribution and their description is given to the learner.
This task can be solved efficiently on quantum computers given that $f$ admits an efficient Fourier series expansion.
We can efficiently estimate $\mathrm{Tr}[e^{-iHt}\rho]$ on quantum computers for different $t$ and perform linear regression using the estimated values as features to construct a model that accurately predicts $y$.
On the other hand, we argue that it is hard to solve this task on classical computers for certain $f$'s and distributions over Hamiltonians and states.
Moreover, we experimentally construct the model using superconducting qubits by using an experimentally friendly technique that does not require controlled-$e^{-iHt}$ operation to estimate $\mathrm{Tr}[e^{-iHt}\rho]$, which is challenging for current quantum devices with limited connectivity and coherence time.
With up to 40-qubit experiments, we demonstrate that the model can be trained and provide accurate predictions even under the intrinsic decoherence and statistical error of the quantum devices.
Our work paves the way for finding practical ML tasks that are advantageous to perform on quantum computers.

\begin{figure}
    \centering
    \includegraphics[width=\linewidth]{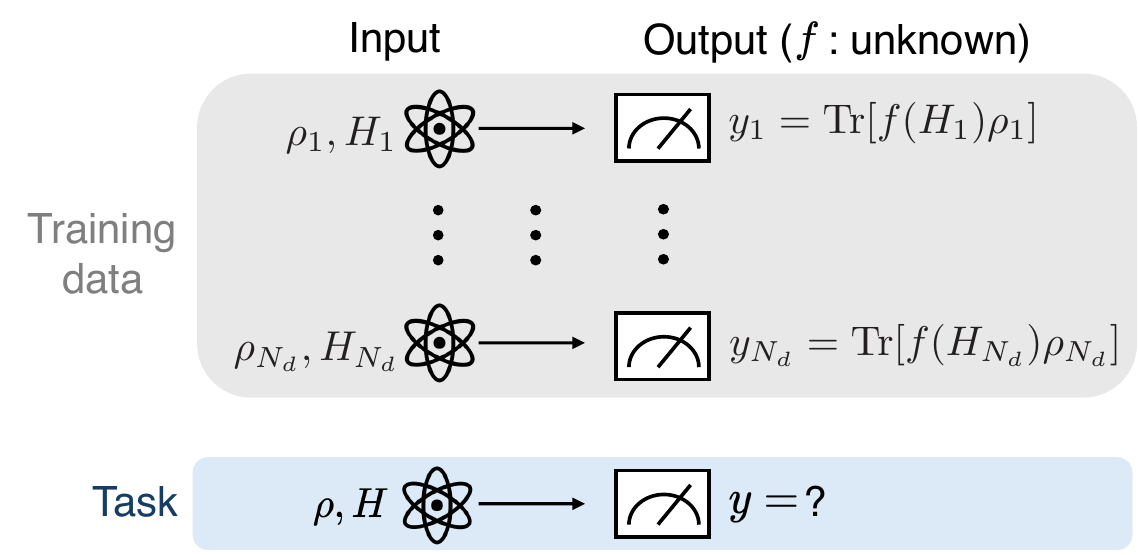}
    \caption{Overview of the ML task proposed in this work. 
    A learner is given a dataset \(\{(H_j, \rho_j, y_j)\}_{j=1}^{N_d}\), consisting of \(n\)-qubit Hamiltonians \(H_j\), \(n\)-qubit states \(\rho_j\), and real-valued labels \(y_j\), where $\rho_j$ and $H_j$ are independent samples from a probability distribution $p(\rho,H)$ and $y_j$ is promised to satisfy $y_j = \mathrm{Tr}[f(H_j)\rho_j]$ for some unknown function $f$. The task of a learner is to predict the value of $y$ for a sample $(\rho, H)$ drawn from $p(\rho, H)$ with a small error.}
    \label{fig:overview}
\end{figure}

\section{Theory}

We first describe the learning task we consider in this work.
We are given a dataset $\{(H_j, \rho_j, y_j)\}_{j=1}^{N_d}$ consisting of $n$-qubit Hamiltonians $H_j$, $n$-qubit states $\rho_j$, and real numbers $y_j$.
Without loss of generality, we assume $\|H\|\leq C$, and the Hamiltonians and states are drawn independently from a probability distribution $p(H,\rho)$.
Importantly, here we assume $p(H,\rho)$ only has support on Hamiltonians and states that admit efficient description, by which we mean one needs only $\mathrm{poly}(n)$ quantum computational resources to conduct Hamiltonian simulation and prepare states.
The real numbers $y$ are promised to be generated via
\begin{align}
    y = \Tr \left[f(H)\rho\right],
\end{align}
where $f:[-C, C]\to\mathbb{R}$ is an unknown function.
The task is to predict $y$ for new Hamiltonians $H$ and states $\rho$ drawn from $p(H,\rho)$.
More precisely, we wish to construct a model $g(H,\rho;\bm{w})$ with trainable parameter $\bm{w}$ such that the expected mean squared loss,
\begin{align}
    R(\bm{w}) = \mathbb{E}_{H, \rho \sim p(H, \rho)}\left[L(\mathrm{Tr}[f(H)\rho], g(H,\rho;\bm{w}))\right],
\end{align}
where $L(y, y')=(y-y')^2$, is small.

This task is easy for learners equipped with quantum computers with the following strategy.
First, we define a feature vector $\bm{x}(H,\rho)=(x_0~\cdots~x_{2K})\in \mathbb{R}^{2K+1}$ for a Hamiltonian $H$ and a state $\rho$ as
\begin{align}\label{eq:fourier-feature}
    x_{k}(H,\rho) &= \begin{cases}
    x_{\mathrm{cos}, k/2} & \text{(even $k$)} \\
    x_{\mathrm{sin}, (k+1)/2} & \text{(odd $k$)} %
    \end{cases},
\end{align}
for $k=0$ to $2K$, where
\begin{align}
    x_{\mathrm{cos}, l} &= \mathrm{Re}\left[\mathrm{Tr}\left(e^{-il\pi H/C}\rho\right)\right] = \mathrm{Tr}\left[\cos\left(k\pi H/C\right)\rho\right] \\
    x_{\mathrm{sin}, l} &= \mathrm{Im}\left[\mathrm{Tr}\left(e^{-il\pi H/C}\rho\right)\right]
    = -\mathrm{Tr}\left[\sin\left(k\pi H/C\right)\rho\right]
\end{align}
These features, which we refer to as Hamiltonian Fourier features, are efficiently computable on a quantum computer using, for example, Hadamard tests. 
Using $\bm{x}(H,\rho)$, we define a linear model
\begin{align}\label{eq:model}
    g(H,\rho;\bm{w}) = \bm{w}\cdot\bm{x}(H,\rho),
\end{align}
where $\bm{w}\in\mathbb{R}^{2K+1}$ are the trainable parameters.
We can see that this model can express a wide range of functions on $[-C, C]$, noting that $g(H, \rho;\bm{w})$ is essentially a Fourier series.
It is well known that one can construct a convergent Fourier series for piecewise continuous functions.
We can thus expect the model \eqref{eq:model} to be a ``good'' model for this task.
More formally, we can prove the following.

\begin{theorem}\label{thm:quantum-easiness}
    Let $ g(H,\rho;\bm{w}) = \sum_{k=0}^K w_{k}x_{k} $ where $x_k$ is defined as in Eq.~\eqref{eq:fourier-feature}, $ L(y, y') = |y - y'|^2 $ be the loss function, and $ f:[-C, C]\to\mathbb{R} $ be a function that admits a Fourier series expansion in the sense that there exists $\{c_k\}_{k=0}^{2K}\in \mathbb{R}^{2K+1}$ such that $|f(x) - (\sum_{k=0}^K c_{2k}\cos(k\pi x/C)+\sum_{k=0}^{K-1} c_{2k+1}\sin(k\pi x/C))|\leq \varepsilon_K$ for all $x\in [-C, C]$ and $\sum_{k=0}^K c_k^2 \leq W^2$ for some constant $\varepsilon_K>0$ and $W=\frac{1}{2C}\int_{-C}^C |f(x)|^2 dx$. 
    Define the expected loss as $R(\bm{w}) = \mathbb{E}_{H,\rho \sim p(H,\rho)}[L(\mathrm{Tr}[f(H)\rho], g(H,\rho;\bm{w}))] $, where $ p(H, \rho) $ is a distribution over Hamiltonians and states satisfying $ \|H\| \leq C $.
    For a set of $ N_d $ samples $ \{(H_i,\rho_i)\}_{i=1}^{N_d} $ drawn independently from $ p(H,\rho) $, with probability at least $ 1-\delta $, the expected loss satisfies
    \begin{multline}\label{eq:expected-loss-bound}
    R(\bm{w}^*) \leq \varepsilon_K^2 
    + 4(\sqrt{2K+1} W + \|f\|_\infty )W 
    \sqrt{\frac{2K+1}{N_d}} \\
    + 3(\sqrt{2K+1} W + \|f\|_\infty )^2 
    \sqrt{\frac{\log \frac{2}{\delta}}{2N_d}},
\end{multline}
    where $ \bm{w}^* $ is the minimizer of the empirical loss $ \hat{R}(\bm{w}) = \frac{1}{N_d} \sum_{j=1}^{N_d} L(\mathrm{Tr}[f(H_j)\rho_j], g(H_j,\rho_j;\bm{w})) $ under the constraint $\bm{w}^T\bm{w}\leq W^2$.
\end{theorem}

Its proof can be found in the Appendix \ref{sec:proof-quantum-easiness}.
The second and third terms in \eqref{eq:expected-loss-bound} can be made arbitrarily small by setting, e.g.,  $N_d=\mathcal{O}(K^2W^4\|f\|_\infty^4)$.
The first term, $\varepsilon_K$, can be made small by taking sufficiently large $K$.
It depends on the class of functions that we wish to learn.
For example, if we restrict $f$ to be Lipschitz continuous, then $\varepsilon_K \leq \mathcal{O}\left(\frac{\log K}{K}\right)$ for all $x\in[-C,C]$ \cite{SalemZygmund1946}.
In this case, we can set $K = \mathcal{O}\left(\log( 1/\varepsilon)/\varepsilon\right)$ to ensure $\varepsilon_K \leq \varepsilon$ for an error $\varepsilon$ we wish to achieve. 
These arguments imply $N_d=\mathcal{O}((W\|f\|_\infty\log(1/\varepsilon)/\varepsilon)^4)$ data are sufficient to construct a model that achieves expected loss less than $\varepsilon$.
This shows that the task at hand is provably easy on quantum computers, given that the unknown function $f$ admits an efficient Fourier series expansion. 

The above discussion can be summarized as the following corollary.
\begin{corollary}
\label{thm:parameter-setting}
Let the symbols be defined as in Theorem \ref{thm:quantum-easiness} but assume $f$ to be Lipschitz continuous.
Then, for any \(\varepsilon > 0\), we can take $K = \mathcal{O}\left(\frac{\log(1/\varepsilon)}{\varepsilon}\right)$ and $N_d = \mathcal{O} \left(\left(\frac{W\|f\|_\infty \log(1/\varepsilon)}{\varepsilon}\right)^4\right)$
to achieve $R(\bm{w}^*) \le \varepsilon$ with probability at least \(1-\delta\).
\end{corollary}

While the above analysis assumes noiseless estimation of the features $\bm{x}(H,\rho)$, practical quantum hardware inevitably introduces statistical fluctuations due to finite measurement shots.
In Appendix~\ref{sec:proof-statistical-easy}, we show that the learnability of the proposed task remains under such statistical noise, by bounding the deviation of the expected loss in the presence of measurement errors.
Specifically, we prove that even when each feature is estimated with bounded additive noise, one can still construct a model whose expected loss is arbitrarily small, provided that the number of measurement shots scales polynomially with the desired accuracy.
This guarantees that the proposed quantum learning scheme is robust against practical imperfections.
 
The classical hardness is clear from the following argument.
For example, if $f(H)=\cos (t H)$ or $f(H)=\sin (t H)$, the problem is essentially equivalent to the quantum simulation problem, which is strongly believed to be a hard task for classical computers to perform for $t=\mathrm{poly}(n)$ \cite{feynman1986quantum, cifuentes2024quantum}.
We leave it as an open question to show this rigorously.
The challenge here is that we need to identify a family of Hamiltonians for which it is hard to estimate $\Tr[f(H)\rho]$ on \textit{average}.
As usual theoretical arguments about the complexity of Hamiltonian simulations refer to worst-case complexity, we cannot straightforwardly show the hardness.
However, we believe it is possible to show classical hardness by, for example, assuming the average case classical hardness of integer factoring and defining a machine learning task like the one used in \cite{Liu2021-pe}; we can rely on the history state Hamiltonians of quantum circuits of Shor's algorithm to fit it in our machine learning task.
We leave it as an interesting topic for future research.

\section{Experimental Demonstration}

We now proceed to demonstrate the successful usage of current quantum computers for the proposed ML task by conducting an experiment on the IBM Quantum platform.

\subsection{Strategy to measure features}
The central challenge in experimentally constructing the model \eqref{eq:model} is the computation of the features $\bm{x}(H, \rho)$.
A popular technique for estimating the values is the Hadamard test, which involves controlled time evolution.
This, however, typically results in a complex quantum circuit that exceeds the capabilities of current quantum hardware.
Thus, we adopt the strategy that works for certain Hamiltonians from \cite{kyriienko2020quantum} in this work.
Let us assume $\rho=\ket{\psi}\bra{\psi}$ for some $\ket{\psi}$.
We make use of an eigenstate $\ket{\psi_\text{ref}}$, corresponding to an eigenvalue $\lambda_{\text{ref}}$ of $H$, and assume that $\braket{\psi|\psi_\text{ref}} = 0$.
Note that we assume we can find this reference eigenstate classically in polynomial time with respect to the system size.
It is possible for, e.g., the Heisenberg model or quantum chemistry Hamiltonians, where the system has symmetry in particle numbers.
Further, define $\ket{\psi_{\pm}} = \frac{\ket{\psi_\text{ref}}\pm\ket{\psi}}{\sqrt{2}}$, $\ket{\psi_{\pm i}} = \frac{\ket{\psi_\text{ref}}\pm i\ket{\psi}}{\sqrt{2}}$, $w_{\pm} = |\braket{\psi_{\pm}|e^{-iHt}|\psi_{+}}|^2$, and $w_{\pm i} = |\braket{\psi_{\pm i}|e^{-iHt}|\psi_{+}}|^2$.
Then, we have
\begin{align}
    \mathrm{Tr}[e^{-iHt}\rho] &= \left[w_+-w_- + i(w_{+i}-w_{-i})\right]e^{-i\lambda_{\rm ref} t}\label{eq:feature-extractor}
\end{align}
Note that we can measure $w_{\pm}$, and $w_{\pm i}$ on a quantum computer if we can efficiently implement unitaries that prepare $\ket{\psi_{\pm}}$, $\ket{\psi_{\pm i}}$ from an initial state $\ket{0}^{\otimes n}$ as a quantum circuit.
Thus, Eq.~\eqref{eq:feature-extractor} provides an experimentally friendly route to estimate elements of $\bm{x}(H, \rho)$.

\subsection{Demonstration setup}
We demonstrate the method under the following conditions.
For dataset, we generate the one-dimensional Heisenberg model $H=\sum_{m=0}^{n-2} J_{m}\left(X_m X_{m+1}+Y_m Y_{m+1}+Z_m Z_{m+1}\right)$, with randomly selected $\{ J_m \}$ and $n=12, 32, 40$. 
The coefficients $\{ J_m \}$ are randomly drawn from a uniform distribution over $[-1,1]$ for all $m$ and subsequently normalized to satisfy $\sum_m | J_m |=1$. Consequently, all Hamiltonians in the dataset satisfy $\| H \|\leq 3$.
The quantum state $\rho$ is set to $\ket{\psi}\bra{\psi}$ with $\ket{\psi}=\ket{0}^{\otimes n/4}\ket{1}^{\otimes n/2}\ket{0}^{\otimes n/4}$ and fixed for all Hamiltonians.
As the reference eigenstate, we use $\ket{\psi_{\rm ref}} = \ket{0}^{\otimes n}$.
By doing so, we can prepare Greenberger–Horne–Zeilinger (GHZ) states with different phases on $n/2$ qubits and initialize other qubits to $\ket{0}$ to create $\ket{\psi_{\pm}}$ and $\ket{\psi_{\pm i}}$.
We adopt $f(H)=e^{-\beta H}$ with $\beta=1$ as a quantity to be learned.
Note that this choice is merely for demonstration purposes, and $f$ can be any function in practice.
In this work, we classically generate the dataset $\{(H_i, \rho_i, y_i)\}_{i=1}^{N_d}$ for each qubit number $n$ by computing $y_j$ via time-evolving block decimation (TEBD) algorithm \cite{vidal2004efficient} using TeNPy \cite{tenpy2024}. 
The bond dimension is set to $\chi=100$, which we verified to give values essentially equivalent to the exact value with an error less than $10^{-6}$ at 12 qubits.
We use \( N_d = 55 \) Hamiltonian samples, which are divided into training and test sets in an \( 8:2 \) ratio.  

For the model hyperparameters, we set \( K=11 \) and $C=3$.
We approximately implement the time evolution \( e^{-iHt} \) by performing a second-order Trotter expansion after partitioning the Hamiltonian into two commuting parts, \( H_{\rm odd} \) and \( H_{\rm even} \), which consist of interaction terms acting on odd and even bonds, respectively.
As a result, we approximate the time evolution by,
\begin{align}
e^{-iH \Delta t} \approx \left( e^{-i\frac{\Delta t}{2}H_\text{odd}} e^{-i \Delta t H_\text{even}} e^{-i\frac{\Delta t}{2}H_\text{odd}} \right)^{n_{\text{step}}},
\end{align}
for some integer $n_{\rm step}$.
We utilize Qiskit to design and execute circuits on an actual quantum processing unit (QPU), specifically the \textit{ibmq\_marrakesh} device of the IBM Quantum platform \cite{ibmquantum, qiskit2024}.
See Appendix \ref{sec:demo-setup} for more details.
For comparison, we also perform noiseless simulations using the matrix product state (MPS) \cite{vidal2003efficient} simulator with a bond dimension of $\chi=100$. Additionally, for the 12-qubit case, we carry out simulations that compute $e^{-iHt}$ exactly without employing the Trotter expansion.

\subsection{Results}
The regression analyses presented in Fig. \ref{fig:demo-result} and Table~\ref{tab:supp-demo-result} evaluate the predictive accuracy of various models for 12, 32, and 40 qubits.
Each plot shows the relationship between actual and predicted values, with results color-coded for different QPU and simulation setups.
The identity line (black) serves as a reference for ideal predictions.
Regression methods are selected based on their mean squared error (MSE) performance on the training dataset, identified using PyCaret \cite{PyCaret}.
For details of experiments, see Appendix \ref{sec:demo-setup}.

We can observe that the result of the exact simulation in Fig.~\ref{subfig:12qubits} shows near-perfect prediction. 
This implies our hyperparameter settings, $K=11$ and $C=3$, have enough representation capability to express $e^{-\beta H}$.
The MPS simulations also consistently achieve near-perfect $R^2$ values across different qubits, demonstrating that the MPS simulation is sufficiently accurate for this task with the chosen Trotter steps and the bond dimension.
Results from QPU under noisy conditions also align well with the identity line, achieving $R^2$ values exceeding 0.80 in the case of 12 and 32 qubits.
For instance, in the 12-qubit system, the regression model achieved an $R^2$ value of 0.977 with an MSE of $2.10 \times 10^{-3}$ on test data.
These findings highlight that the proposed ML task, which we strongly believe is classically intractable for certain families of Hamiltonians, is solvable even under the effect of intrinsic decoherence and statistical noise of quantum hardware.

On the other hand, the 40-qubit experiment does not provide a meaningful prediction as shown in Fig.~\ref{subfig:40qubits}, possibly due to the device noise.
To investigate the effect of noise, Fig.~\ref{fig:SimvsQPU} presents a scatter plot of the values of experimental $w_{\pm}$ and $w_{\pm i}$ against their theoretical values. 
At 12 qubits (Fig.~\ref{subfig:SimvsQPU_12qubits}), we can see that the QPU produces results that closely align with those of the simulator; however, as the system size increases, the discrepancy between the QPU and the simulator becomes increasingly pronounced.
At 40 qubits (Fig.~\ref{subfig:SimvsQPU_40qubits}), most of the experimental values of $w_{\pm}$ and $w_{\pm i}$ do not show much correlation with the theoretical value, suggesting that the measurement has just returned random results in most of the experiments.
We note that prominent peaks appearing at 1.0 for $w_{+}$ and at 0.5 for $w_{\pm i}$ on the horizontal axis correspond to $t = 0$ ($l=0$).
Since no specialized error mitigation techniques were employed in this experiment, the observed hardware errors could potentially be reduced by applying them \cite{cai2023quantum, endo2021hybrid}.

\begin{figure*}
\centering
\subfloat[12 Qubits\label{subfig:12qubits}]{
    \includegraphics[width=0.3\textwidth]{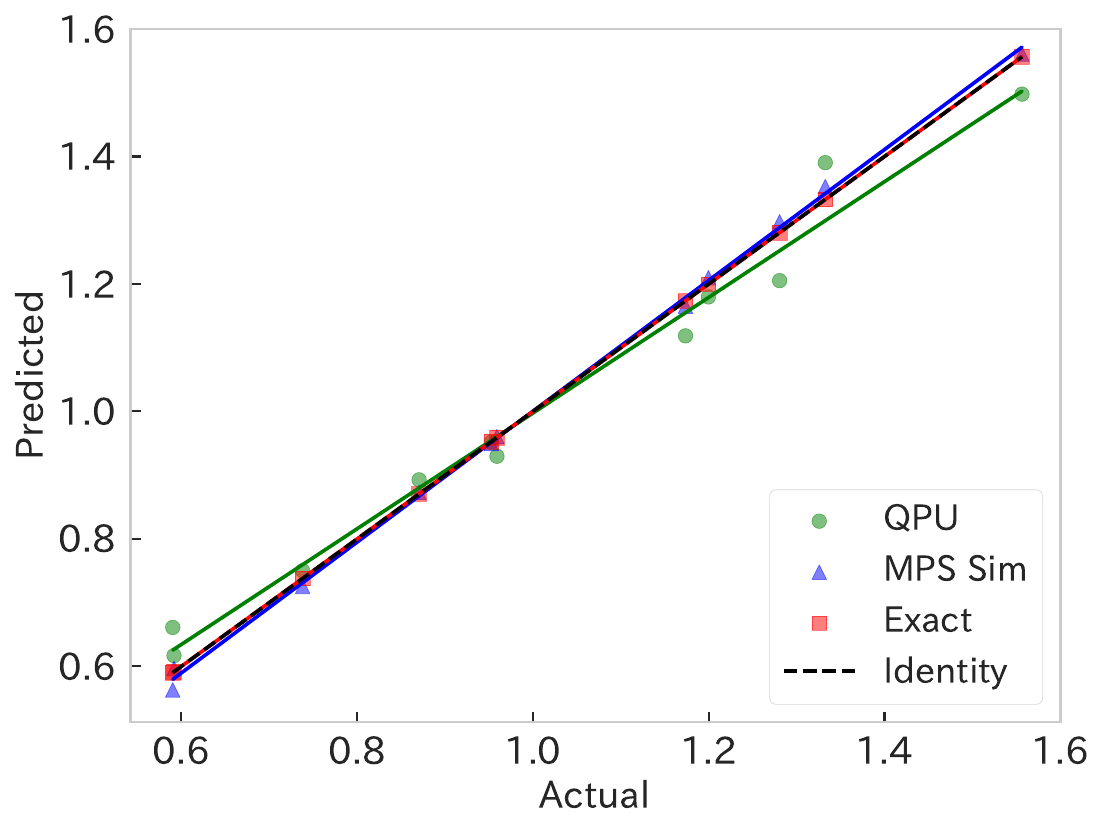}
}
\subfloat[32 Qubits\label{subfig:32qubits}]{
    \includegraphics[width=0.3\textwidth]{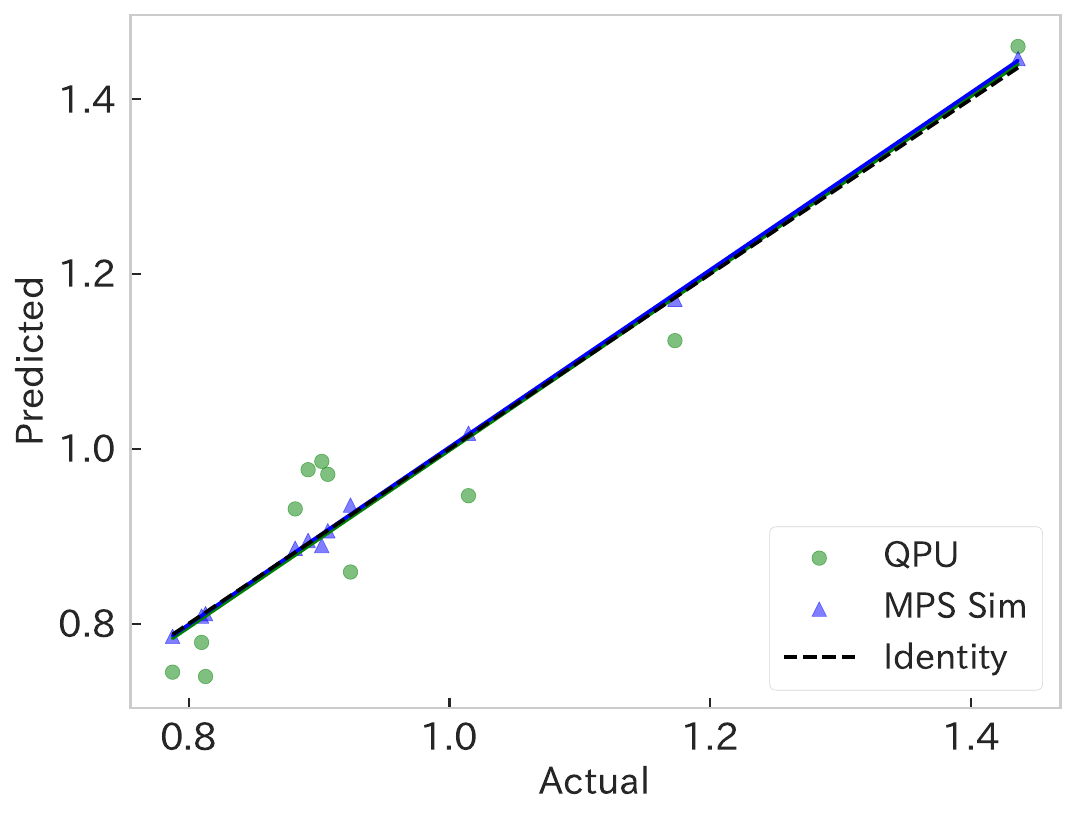}
}
\subfloat[40 Qubits\label{subfig:40qubits}]{
    \includegraphics[width=0.3\textwidth]{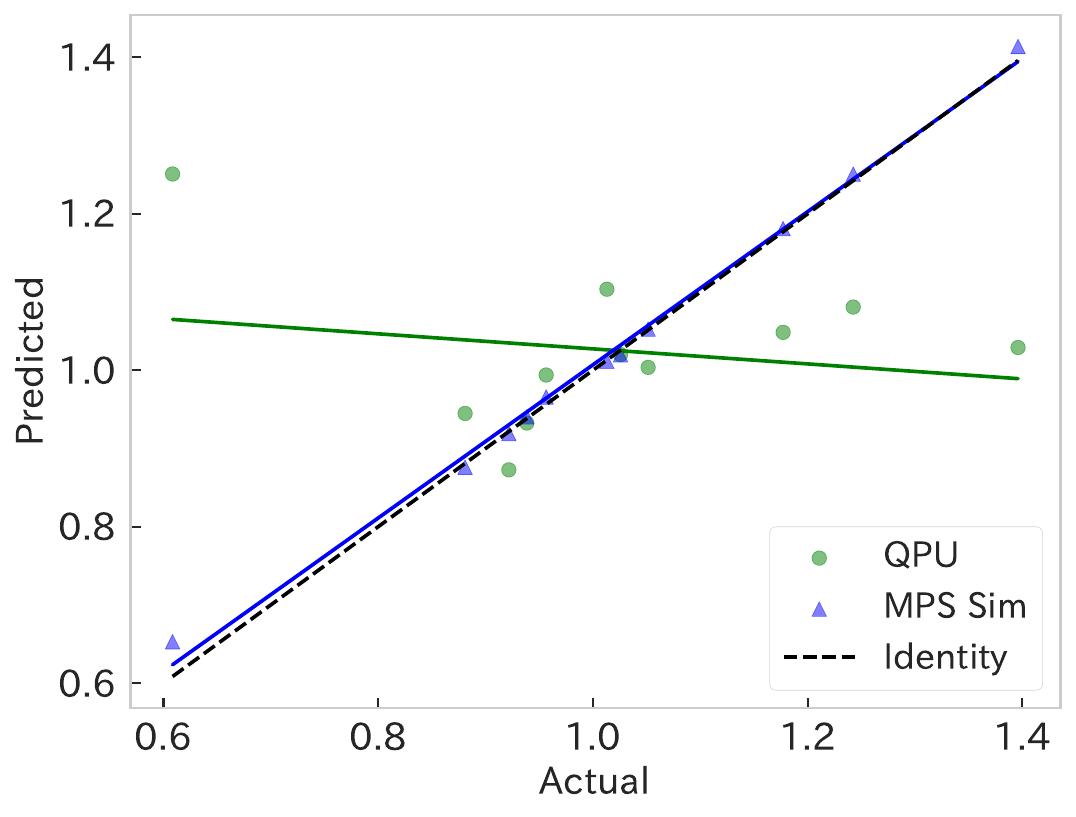}
}
\caption{Demonstration results for 12, 32, and 40 qubits.
    The horizontal axes represent actual $y$ values, while the vertical axes show predicted values. 
    Data points and best-fit lines, which we place for eye-guiding purposes, are color-coded: 
    (green) results from IBM's QPU, \textit{ibmq\_marrakesh},  
    (red) Results from exact matrix exponentiation without the use of Trotter expansion, and
    (blue) Results from the noise-free MPS simulator. Additional details are provided in Table.~\ref{tab:supp-demo-result}.}
\label{fig:demo-result}
\end{figure*}

\begin{table*}[t]
\caption{\label{tab:supp-demo-result}Additional information corresponding to the demonstration results shown in Fig.~\ref{fig:demo-result}. $R^2$ and MSE are calculated on the test dataset.}
\begin{ruledtabular}
\begin{tabular}{ccccc}
\# of qubits & Backend & Regression method & $R^2$ & MSE \\
\hline
12 & Exact & LinearRegression & 1.00 & $1.47\times 10^{-10}$ \\
12 & QPU & BayesianRidge & 0.977 & $2.10\times 10^{-3}$ \\
12 & MPS & BayesianRidge & 0.998 & $1.66\times 10^{-4}$ \\
32 & QPU & ExtraTreesRegressor & 0.889 & $3.71\times 10^{-3}$ \\
32 & MPS & LinearRegression & 0.999 & $4.13\times 10^{-5}$  \\
40 & QPU & ExtraTreesRegressor & $-0.430$ & $5.53\times 10^{-2}$  \\
40 & MPS & BayesianRidge & 0.994 & $2.32\times 10^{-4}$ \\
\end{tabular}
\end{ruledtabular}
\end{table*}

\begin{figure*}[!t]
\centering
\subfloat[12 Qubits\label{subfig:SimvsQPU_12qubits}]{
    \includegraphics[width=0.3\textwidth]{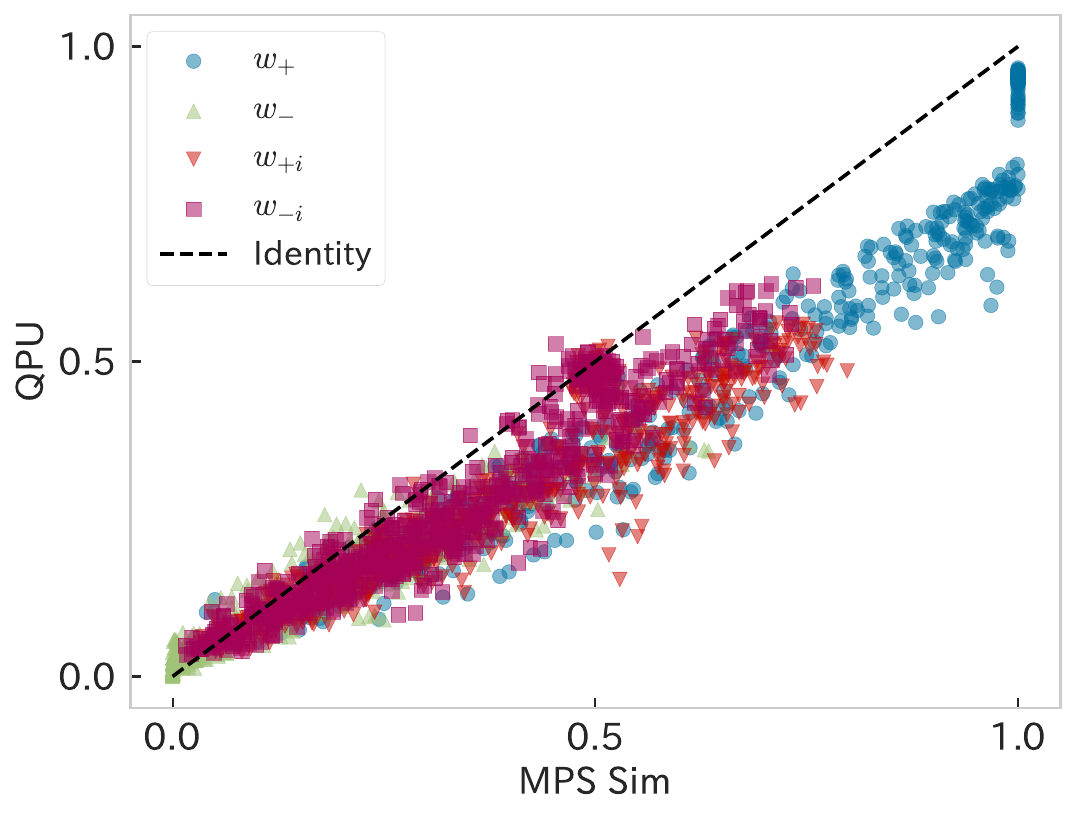}
}
\subfloat[32 Qubits\label{subfig:SimvsQPU_32qubits}]{
    \includegraphics[width=0.3\textwidth]{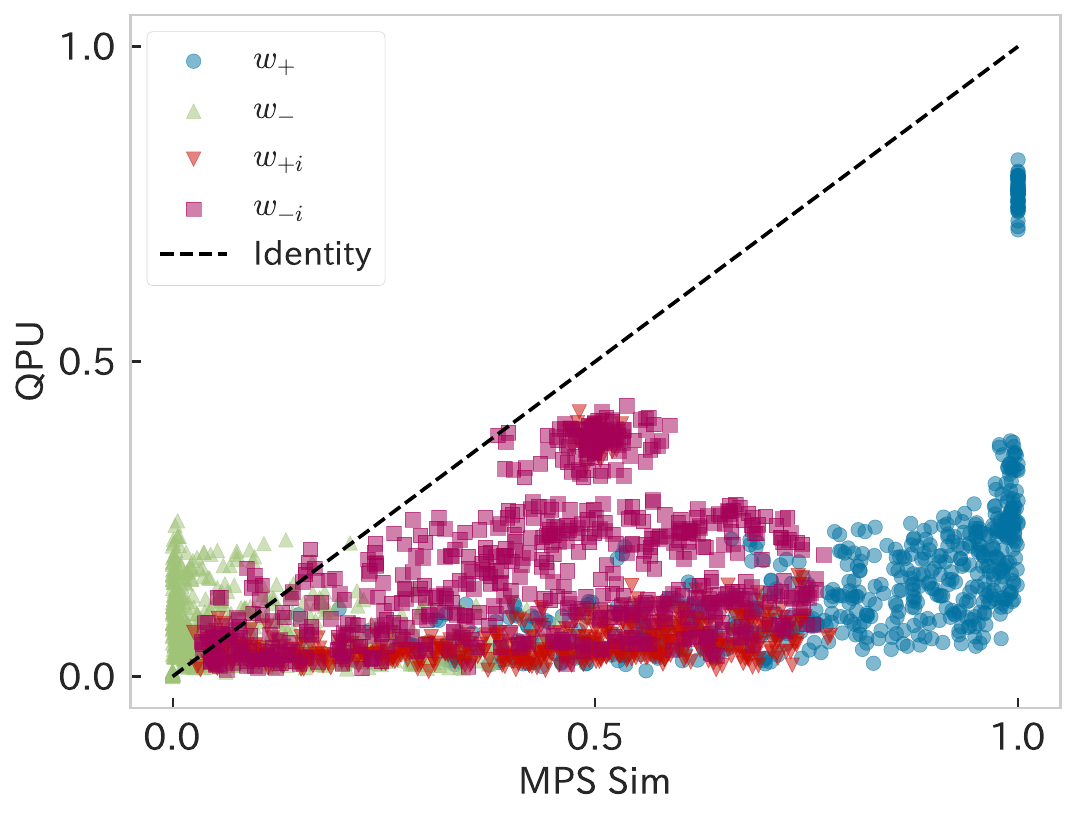}
}
\subfloat[40 Qubits\label{subfig:SimvsQPU_40qubits}]{
    \includegraphics[width=0.3\textwidth]{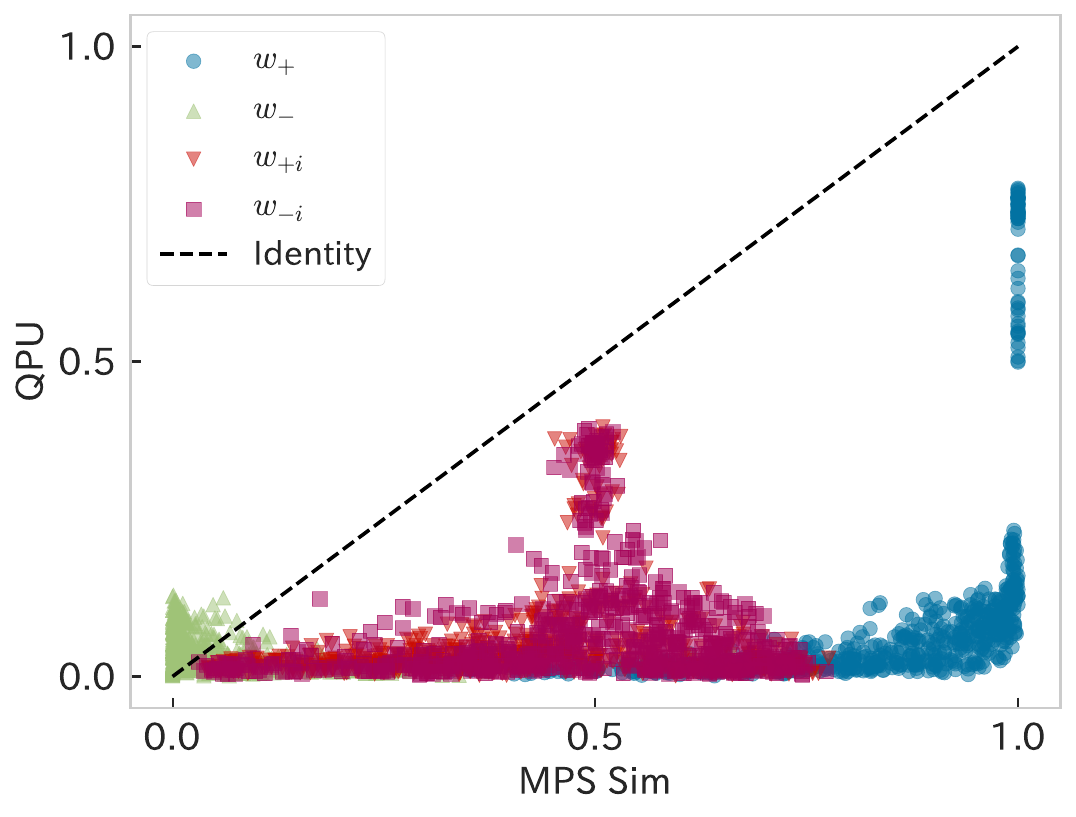}
}
\caption{
Comparison between the QPU and the noiseless MPS simulator, which effectively provides exact simulation results.
The horizontal axis represents values of $w_{\pm}$ and $w_{\pm i}$ in the noiseless MPS simulator, while the vertical axis shows the corresponding values measured on the QPU.}
\label{fig:SimvsQPU}
\end{figure*}

\section{Conclusion}
In this study, we introduced a machine learning task that is provably easy for quantum computers but arguably hard for classical computers. 
By formulating a supervised learning task where the prediction target is $y = \mathrm{Tr}[f(H)\rho]$ for an unknown function $f$ admitting a Fourier series expansion, we provided a separation between quantum and classical computational capabilities. 
The proposed quantum approach to solve the ML task leverages Hamiltonian Fourier features $\bm{x}(H)$ that can be obtained with quantum computers and integrates these features into a linear regression model. 
Theoretical analysis confirms that the task is efficiently solvable on quantum computers while being hard for classical computers under certain conditions.
Our experimental implementation using superconducting qubits further substantiated the feasibility of the ML model. 
The technique that avoids the Hadamard tests enabled the demonstration of the model construction on IBM’s quantum processors, even under noisy conditions. 

We finally discuss possible future directions.
First, an interesting direction is to explore the broader applicability of the Hamiltonian Fourier feature model introduced in this study.
While we have theoretically established its effectiveness for the specific learning task considered here, it remains an open and intriguing question whether this approach can solve other ML tasks that take Hamiltonians as explanatory variables.
Second, a challenge that remains is to formalize the classical hardness of the proposed task, particularly under average-case complexity assumptions. \\

\noindent\textbf{Data availability:} The code used in this work is available online at \url{https://github.com/YutoMorohoshi/fourier_learning_ibm}. \\

\section*{Acknowledgments}
K.M. thanks Hayata Yamasaki for a fruitful discussion at the very early stage of this project.
We acknowledge the use of IBM Quantum services for this work. 
Y.M. and H.M. thank Ryo Watanabe for the useful discussion about tensor networks.
Y.M. is supported by JST SPRING under Grant Number JPMJSP2138.
H.M. is supported by JST COI-NEXT program Grant Numbers JPMJPF2014.
This work is supported by MEXT Quantum Leap Flagship Program (MEXT-QLEAP) Grant Nos. JPMXS0120319794 and JPMXS0118067394, JST COI-NEXT Grant No. JPMJPF2014, and NEDO Grant No. JPNP20017.
K.M. is supported by JST FOREST Grant No. JPMJFR232Z, JSPS KAKENHI Grant No. 23H03819, 24K16980 and JST CREST Grant No. JPMJCR24I4.

\bibliographystyle{apsrev4-2}
\bibliography{main.bib}

\begin{thebibliography}{24}%
\makeatletter
\providecommand \@ifxundefined [1]{%
 \@ifx{#1\undefined}
}%
\providecommand \@ifnum [1]{%
 \ifnum #1\expandafter \@firstoftwo
 \else \expandafter \@secondoftwo
 \fi
}%
\providecommand \@ifx [1]{%
 \ifx #1\expandafter \@firstoftwo
 \else \expandafter \@secondoftwo
 \fi
}%
\providecommand \natexlab [1]{#1}%
\providecommand \enquote  [1]{``#1''}%
\providecommand \bibnamefont  [1]{#1}%
\providecommand \bibfnamefont [1]{#1}%
\providecommand \citenamefont [1]{#1}%
\providecommand \href@noop [0]{\@secondoftwo}%
\providecommand \href [0]{\begingroup \@sanitize@url \@href}%
\providecommand \@href[1]{\@@startlink{#1}\@@href}%
\providecommand \@@href[1]{\endgroup#1\@@endlink}%
\providecommand \@sanitize@url [0]{\catcode `\\12\catcode `\$12\catcode `\&12\catcode `\#12\catcode `\^12\catcode `\_12\catcode `\%12\relax}%
\providecommand \@@startlink[1]{}%
\providecommand \@@endlink[0]{}%
\providecommand \url  [0]{\begingroup\@sanitize@url \@url }%
\providecommand \@url [1]{\endgroup\@href {#1}{\urlprefix }}%
\providecommand \urlprefix  [0]{URL }%
\providecommand \Eprint [0]{\href }%
\providecommand \doibase [0]{https://doi.org/}%
\providecommand \selectlanguage [0]{\@gobble}%
\providecommand \bibinfo  [0]{\@secondoftwo}%
\providecommand \bibfield  [0]{\@secondoftwo}%
\providecommand \translation [1]{[#1]}%
\providecommand \BibitemOpen [0]{}%
\providecommand \bibitemStop [0]{}%
\providecommand \bibitemNoStop [0]{.\EOS\space}%
\providecommand \EOS [0]{\spacefactor3000\relax}%
\providecommand \BibitemShut  [1]{\csname bibitem#1\endcsname}%
\let\auto@bib@innerbib\@empty
\bibitem [{\citenamefont {Cerezo}\ \emph {et~al.}(2022)\citenamefont {Cerezo}, \citenamefont {Verdon}, \citenamefont {Huang}, \citenamefont {Cincio},\ and\ \citenamefont {Coles}}]{Cerezo2022-xy}%
  \BibitemOpen
  \bibfield  {author} {\bibinfo {author} {\bibfnamefont {M.}~\bibnamefont {Cerezo}}, \bibinfo {author} {\bibfnamefont {G.}~\bibnamefont {Verdon}}, \bibinfo {author} {\bibfnamefont {H.-Y.}\ \bibnamefont {Huang}}, \bibinfo {author} {\bibfnamefont {L.}~\bibnamefont {Cincio}},\ and\ \bibinfo {author} {\bibfnamefont {P.~J.}\ \bibnamefont {Coles}},\ }\href {https://doi.org/10.1038/s43588-022-00311-3} {\bibfield  {journal} {\bibinfo  {journal} {Nature computational science}\ }\textbf {\bibinfo {volume} {2}},\ \bibinfo {pages} {567} (\bibinfo {year} {2022})}\BibitemShut {NoStop}%
\bibitem [{\citenamefont {Biamonte}\ \emph {et~al.}(2017)\citenamefont {Biamonte}, \citenamefont {Wittek}, \citenamefont {Pancotti}, \citenamefont {Rebentrost}, \citenamefont {Wiebe},\ and\ \citenamefont {Lloyd}}]{Biamonte2017-ht}%
  \BibitemOpen
  \bibfield  {author} {\bibinfo {author} {\bibfnamefont {J.}~\bibnamefont {Biamonte}}, \bibinfo {author} {\bibfnamefont {P.}~\bibnamefont {Wittek}}, \bibinfo {author} {\bibfnamefont {N.}~\bibnamefont {Pancotti}}, \bibinfo {author} {\bibfnamefont {P.}~\bibnamefont {Rebentrost}}, \bibinfo {author} {\bibfnamefont {N.}~\bibnamefont {Wiebe}},\ and\ \bibinfo {author} {\bibfnamefont {S.}~\bibnamefont {Lloyd}},\ }\href {https://doi.org/10.1038/nature23474} {\bibfield  {journal} {\bibinfo  {journal} {Nature}\ }\textbf {\bibinfo {volume} {549}},\ \bibinfo {pages} {195} (\bibinfo {year} {2017})}\BibitemShut {NoStop}%
\bibitem [{\citenamefont {Cerezo}\ \emph {et~al.}(2021)\citenamefont {Cerezo}, \citenamefont {Arrasmith}, \citenamefont {Babbush}, \citenamefont {Benjamin}, \citenamefont {Endo}, \citenamefont {Fujii}, \citenamefont {McClean}, \citenamefont {Mitarai}, \citenamefont {Yuan}, \citenamefont {Cincio},\ and\ \citenamefont {Coles}}]{Cerezo2021-gk}%
  \BibitemOpen
  \bibfield  {author} {\bibinfo {author} {\bibfnamefont {M.}~\bibnamefont {Cerezo}}, \bibinfo {author} {\bibfnamefont {A.}~\bibnamefont {Arrasmith}}, \bibinfo {author} {\bibfnamefont {R.}~\bibnamefont {Babbush}}, \bibinfo {author} {\bibfnamefont {S.~C.}\ \bibnamefont {Benjamin}}, \bibinfo {author} {\bibfnamefont {S.}~\bibnamefont {Endo}}, \bibinfo {author} {\bibfnamefont {K.}~\bibnamefont {Fujii}}, \bibinfo {author} {\bibfnamefont {J.~R.}\ \bibnamefont {McClean}}, \bibinfo {author} {\bibfnamefont {K.}~\bibnamefont {Mitarai}}, \bibinfo {author} {\bibfnamefont {X.}~\bibnamefont {Yuan}}, \bibinfo {author} {\bibfnamefont {L.}~\bibnamefont {Cincio}},\ and\ \bibinfo {author} {\bibfnamefont {P.~J.}\ \bibnamefont {Coles}},\ }\href {https://doi.org/10.1038/s42254-021-00348-9} {\bibfield  {journal} {\bibinfo  {journal} {Nature Reviews Physics}\ }\textbf {\bibinfo {volume} {3}},\ \bibinfo {pages} {625} (\bibinfo {year} {2021})}\BibitemShut {NoStop}%
\bibitem [{\citenamefont {Servedio}\ and\ \citenamefont {Gortler}(2004)}]{Servedio2004Equivalences}%
  \BibitemOpen
  \bibfield  {author} {\bibinfo {author} {\bibfnamefont {R.~A.}\ \bibnamefont {Servedio}}\ and\ \bibinfo {author} {\bibfnamefont {S.~J.}\ \bibnamefont {Gortler}},\ }\href {https://doi.org/10.1137/S0097539704412910} {\bibfield  {journal} {\bibinfo  {journal} {SIAM Journal on Computing}\ }\textbf {\bibinfo {volume} {33}},\ \bibinfo {pages} {1067} (\bibinfo {year} {2004})}\BibitemShut {NoStop}%
\bibitem [{\citenamefont {Kearns}\ and\ \citenamefont {Valiant}(1994)}]{Kearns1994}%
  \BibitemOpen
  \bibfield  {author} {\bibinfo {author} {\bibfnamefont {M.}~\bibnamefont {Kearns}}\ and\ \bibinfo {author} {\bibfnamefont {L.}~\bibnamefont {Valiant}},\ }\href {https://doi.org/10.1145/174644.174647} {\bibfield  {journal} {\bibinfo  {journal} {J. ACM}\ }\textbf {\bibinfo {volume} {41}},\ \bibinfo {pages} {67–95} (\bibinfo {year} {1994})}\BibitemShut {NoStop}%
\bibitem [{\citenamefont {Liu}\ \emph {et~al.}(2021)\citenamefont {Liu}, \citenamefont {Arunachalam},\ and\ \citenamefont {Temme}}]{Liu2021-pe}%
  \BibitemOpen
  \bibfield  {author} {\bibinfo {author} {\bibfnamefont {Y.}~\bibnamefont {Liu}}, \bibinfo {author} {\bibfnamefont {S.}~\bibnamefont {Arunachalam}},\ and\ \bibinfo {author} {\bibfnamefont {K.}~\bibnamefont {Temme}},\ }\href {https://doi.org/10.1038/s41567-021-01287-z} {\bibfield  {journal} {\bibinfo  {journal} {Nature physics}\ }\textbf {\bibinfo {volume} {17}},\ \bibinfo {pages} {1013} (\bibinfo {year} {2021})}\BibitemShut {NoStop}%
\bibitem [{\citenamefont {Schuld}\ and\ \citenamefont {Killoran}(2019)}]{Schuld2019Quantum}%
  \BibitemOpen
  \bibfield  {author} {\bibinfo {author} {\bibfnamefont {M.}~\bibnamefont {Schuld}}\ and\ \bibinfo {author} {\bibfnamefont {N.}~\bibnamefont {Killoran}},\ }\href {https://doi.org/10.1103/PhysRevLett.122.040504} {\bibfield  {journal} {\bibinfo  {journal} {Physical Review Letters}\ }\textbf {\bibinfo {volume} {122}},\ \bibinfo {pages} {040504} (\bibinfo {year} {2019})}\BibitemShut {NoStop}%
\bibitem [{\citenamefont {Havl{\'{i}}{\v{c}}ek}\ \emph {et~al.}(2019)\citenamefont {Havl{\'{i}}{\v{c}}ek}, \citenamefont {C{\'{o}}rcoles}, \citenamefont {Temme}, \citenamefont {Harrow}, \citenamefont {Kandala}, \citenamefont {Chow},\ and\ \citenamefont {Gambetta}}]{Havlicek2019Supervised}%
  \BibitemOpen
  \bibfield  {author} {\bibinfo {author} {\bibfnamefont {V.}~\bibnamefont {Havl{\'{i}}{\v{c}}ek}}, \bibinfo {author} {\bibfnamefont {A.~D.}\ \bibnamefont {C{\'{o}}rcoles}}, \bibinfo {author} {\bibfnamefont {K.}~\bibnamefont {Temme}}, \bibinfo {author} {\bibfnamefont {A.~W.}\ \bibnamefont {Harrow}}, \bibinfo {author} {\bibfnamefont {A.}~\bibnamefont {Kandala}}, \bibinfo {author} {\bibfnamefont {J.~M.}\ \bibnamefont {Chow}},\ and\ \bibinfo {author} {\bibfnamefont {J.~M.}\ \bibnamefont {Gambetta}},\ }\href {https://doi.org/10.1038/s41586-019-0980-2} {\bibfield  {journal} {\bibinfo  {journal} {Nature}\ }\textbf {\bibinfo {volume} {567}},\ \bibinfo {pages} {209} (\bibinfo {year} {2019})}\BibitemShut {NoStop}%
\bibitem [{\citenamefont {Yamasaki}\ \emph {et~al.}(2023)\citenamefont {Yamasaki}, \citenamefont {Isogai},\ and\ \citenamefont {Murao}}]{YamasakiIsogaiMurao2023}%
  \BibitemOpen
  \bibfield  {author} {\bibinfo {author} {\bibfnamefont {H.}~\bibnamefont {Yamasaki}}, \bibinfo {author} {\bibfnamefont {N.}~\bibnamefont {Isogai}},\ and\ \bibinfo {author} {\bibfnamefont {M.}~\bibnamefont {Murao}},\ }\Eprint {https://arxiv.org/abs/2312.03057} {arXiv:2312.03057 [quant-ph]}  (\bibinfo {year} {2023})\BibitemShut {NoStop}%
\bibitem [{\citenamefont {Gyurik}\ and\ \citenamefont {Dunjko}(2023)}]{Gyurik2023Exponential}%
  \BibitemOpen
  \bibfield  {author} {\bibinfo {author} {\bibfnamefont {C.}~\bibnamefont {Gyurik}}\ and\ \bibinfo {author} {\bibfnamefont {V.}~\bibnamefont {Dunjko}},\ }\Eprint {https://arxiv.org/abs/2306.16028} {arXiv:2306.16028 [quant-ph]}  (\bibinfo {year} {2023})\BibitemShut {NoStop}%
\bibitem [{\citenamefont {Molteni}\ \emph {et~al.}(2024)\citenamefont {Molteni}, \citenamefont {Gyurik},\ and\ \citenamefont {Dunjko}}]{MolteniGyurikDunjko2024}%
  \BibitemOpen
  \bibfield  {author} {\bibinfo {author} {\bibfnamefont {R.}~\bibnamefont {Molteni}}, \bibinfo {author} {\bibfnamefont {C.}~\bibnamefont {Gyurik}},\ and\ \bibinfo {author} {\bibfnamefont {V.}~\bibnamefont {Dunjko}},\ }\Eprint {https://arxiv.org/abs/2405.02027} {arXiv:2405.02027 [quant-ph]}  (\bibinfo {year} {2024})\BibitemShut {NoStop}%
\bibitem [{\citenamefont {Salem}\ and\ \citenamefont {Zygmund}(1946)}]{SalemZygmund1946}%
  \BibitemOpen
  \bibfield  {author} {\bibinfo {author} {\bibfnamefont {R.}~\bibnamefont {Salem}}\ and\ \bibinfo {author} {\bibfnamefont {A.}~\bibnamefont {Zygmund}},\ }\href {https://doi.org/https://doi.org/10.1090/S0002-9947-1946-0015538-0} {\bibfield  {journal} {\bibinfo  {journal} {Trans. Amer. Math. Soc.}\ }\textbf {\bibinfo {volume} {59}},\ \bibinfo {pages} {14} (\bibinfo {year} {1946})}\BibitemShut {NoStop}%
\bibitem [{\citenamefont {Feynman}(1986)}]{feynman1986quantum}%
  \BibitemOpen
  \bibfield  {author} {\bibinfo {author} {\bibfnamefont {R.~P.}\ \bibnamefont {Feynman}},\ }\href@noop {} {\bibfield  {journal} {\bibinfo  {journal} {Found. Phys.}\ }\textbf {\bibinfo {volume} {16}},\ \bibinfo {pages} {507} (\bibinfo {year} {1986})}\BibitemShut {NoStop}%
\bibitem [{\citenamefont {Cifuentes}\ \emph {et~al.}(2024)\citenamefont {Cifuentes}, \citenamefont {Wang}, \citenamefont {Silva}, \citenamefont {Berta},\ and\ \citenamefont {Aolita}}]{cifuentes2024quantum}%
  \BibitemOpen
  \bibfield  {author} {\bibinfo {author} {\bibfnamefont {S.}~\bibnamefont {Cifuentes}}, \bibinfo {author} {\bibfnamefont {S.}~\bibnamefont {Wang}}, \bibinfo {author} {\bibfnamefont {T.~L.}\ \bibnamefont {Silva}}, \bibinfo {author} {\bibfnamefont {M.}~\bibnamefont {Berta}},\ and\ \bibinfo {author} {\bibfnamefont {L.}~\bibnamefont {Aolita}},\ }\href@noop {} {\bibfield  {journal} {\bibinfo  {journal} {arXiv preprint arXiv:2410.13937}\ } (\bibinfo {year} {2024})}\BibitemShut {NoStop}%
\bibitem [{\citenamefont {Kyriienko}(2020)}]{kyriienko2020quantum}%
  \BibitemOpen
  \bibfield  {author} {\bibinfo {author} {\bibfnamefont {O.}~\bibnamefont {Kyriienko}},\ }\href {https://doi.org/10.1038/s41534-019-0239-9} {\bibfield  {journal} {\bibinfo  {journal} {npj Quantum Information}\ }\textbf {\bibinfo {volume} {6}},\ \bibinfo {pages} {7} (\bibinfo {year} {2020})}\BibitemShut {NoStop}%
\bibitem [{\citenamefont {Vidal}(2004)}]{vidal2004efficient}%
  \BibitemOpen
  \bibfield  {author} {\bibinfo {author} {\bibfnamefont {G.}~\bibnamefont {Vidal}},\ }\href@noop {} {\bibfield  {journal} {\bibinfo  {journal} {Physical review letters}\ }\textbf {\bibinfo {volume} {93}},\ \bibinfo {pages} {040502} (\bibinfo {year} {2004})}\BibitemShut {NoStop}%
\bibitem [{\citenamefont {Hauschild}\ \emph {et~al.}(2024)\citenamefont {Hauschild}, \citenamefont {Unfried}, \citenamefont {Anand}, \citenamefont {Andrews}, \citenamefont {Bintz}, \citenamefont {Borla}, \citenamefont {Divic}, \citenamefont {Drescher}, \citenamefont {Geiger}, \citenamefont {Hefel}, \citenamefont {Hémery}, \citenamefont {Kadow}, \citenamefont {Kemp}, \citenamefont {Kirchner}, \citenamefont {Liu}, \citenamefont {Möller}, \citenamefont {Parker}, \citenamefont {Rader}, \citenamefont {Romen}, \citenamefont {Scalet}, \citenamefont {Schoonderwoerd}, \citenamefont {Schulz}, \citenamefont {Soejima}, \citenamefont {Thoma}, \citenamefont {Wu}, \citenamefont {Zechmann}, \citenamefont {Zweng}, \citenamefont {Mong}, \citenamefont {Zaletel},\ and\ \citenamefont {Pollmann}}]{tenpy2024}%
  \BibitemOpen
  \bibfield  {author} {\bibinfo {author} {\bibfnamefont {J.}~\bibnamefont {Hauschild}}, \bibinfo {author} {\bibfnamefont {J.}~\bibnamefont {Unfried}}, \bibinfo {author} {\bibfnamefont {S.}~\bibnamefont {Anand}}, \bibinfo {author} {\bibfnamefont {B.}~\bibnamefont {Andrews}}, \bibinfo {author} {\bibfnamefont {M.}~\bibnamefont {Bintz}}, \bibinfo {author} {\bibfnamefont {U.}~\bibnamefont {Borla}}, \bibinfo {author} {\bibfnamefont {S.}~\bibnamefont {Divic}}, \bibinfo {author} {\bibfnamefont {M.}~\bibnamefont {Drescher}}, \bibinfo {author} {\bibfnamefont {J.}~\bibnamefont {Geiger}}, \bibinfo {author} {\bibfnamefont {M.}~\bibnamefont {Hefel}}, \bibinfo {author} {\bibfnamefont {K.}~\bibnamefont {Hémery}}, \bibinfo {author} {\bibfnamefont {W.}~\bibnamefont {Kadow}}, \bibinfo {author} {\bibfnamefont {J.}~\bibnamefont {Kemp}}, \bibinfo {author} {\bibfnamefont {N.}~\bibnamefont {Kirchner}}, \bibinfo {author} {\bibfnamefont {V.~S.}\ \bibnamefont {Liu}}, \bibinfo {author} {\bibfnamefont {G.}~\bibnamefont {Möller}},
  \bibinfo {author} {\bibfnamefont {D.}~\bibnamefont {Parker}}, \bibinfo {author} {\bibfnamefont {M.}~\bibnamefont {Rader}}, \bibinfo {author} {\bibfnamefont {A.}~\bibnamefont {Romen}}, \bibinfo {author} {\bibfnamefont {S.}~\bibnamefont {Scalet}}, \bibinfo {author} {\bibfnamefont {L.}~\bibnamefont {Schoonderwoerd}}, \bibinfo {author} {\bibfnamefont {M.}~\bibnamefont {Schulz}}, \bibinfo {author} {\bibfnamefont {T.}~\bibnamefont {Soejima}}, \bibinfo {author} {\bibfnamefont {P.}~\bibnamefont {Thoma}}, \bibinfo {author} {\bibfnamefont {Y.}~\bibnamefont {Wu}}, \bibinfo {author} {\bibfnamefont {P.}~\bibnamefont {Zechmann}}, \bibinfo {author} {\bibfnamefont {L.}~\bibnamefont {Zweng}}, \bibinfo {author} {\bibfnamefont {R.~S.~K.}\ \bibnamefont {Mong}}, \bibinfo {author} {\bibfnamefont {M.~P.}\ \bibnamefont {Zaletel}},\ and\ \bibinfo {author} {\bibfnamefont {F.}~\bibnamefont {Pollmann}},\ }\href {https://doi.org/10.21468/SciPostPhysCodeb.41} {\bibfield  {journal} {\bibinfo  {journal} {SciPost Phys. Codebases}\ ,\
  \bibinfo {pages} {41}} (\bibinfo {year} {2024})}\BibitemShut {NoStop}%
\bibitem [{ibm(2025)}]{ibmquantum}%
  \BibitemOpen
  \href@noop {} {\bibinfo {title} {{IBM Q}uantum}},\ \bibinfo {howpublished} {https://quantum.ibm.com/} (\bibinfo {year} {2025})\BibitemShut {NoStop}%
\bibitem [{\citenamefont {Javadi-Abhari}\ \emph {et~al.}(2024)\citenamefont {Javadi-Abhari}, \citenamefont {Treinish}, \citenamefont {Krsulich}, \citenamefont {Wood}, \citenamefont {Lishman}, \citenamefont {Gacon}, \citenamefont {Martiel}, \citenamefont {Nation}, \citenamefont {Bishop}, \citenamefont {Cross}, \citenamefont {Johnson},\ and\ \citenamefont {Gambetta}}]{qiskit2024}%
  \BibitemOpen
  \bibfield  {author} {\bibinfo {author} {\bibfnamefont {A.}~\bibnamefont {Javadi-Abhari}}, \bibinfo {author} {\bibfnamefont {M.}~\bibnamefont {Treinish}}, \bibinfo {author} {\bibfnamefont {K.}~\bibnamefont {Krsulich}}, \bibinfo {author} {\bibfnamefont {C.~J.}\ \bibnamefont {Wood}}, \bibinfo {author} {\bibfnamefont {J.}~\bibnamefont {Lishman}}, \bibinfo {author} {\bibfnamefont {J.}~\bibnamefont {Gacon}}, \bibinfo {author} {\bibfnamefont {S.}~\bibnamefont {Martiel}}, \bibinfo {author} {\bibfnamefont {P.~D.}\ \bibnamefont {Nation}}, \bibinfo {author} {\bibfnamefont {L.~S.}\ \bibnamefont {Bishop}}, \bibinfo {author} {\bibfnamefont {A.~W.}\ \bibnamefont {Cross}}, \bibinfo {author} {\bibfnamefont {B.~R.}\ \bibnamefont {Johnson}},\ and\ \bibinfo {author} {\bibfnamefont {J.~M.}\ \bibnamefont {Gambetta}},\ }\href {https://doi.org/10.48550/arXiv.2405.08810} {\bibinfo {title} {Quantum computing with {Q}iskit}} (\bibinfo {year} {2024}),\ \Eprint {https://arxiv.org/abs/2405.08810} {arXiv:2405.08810 [quant-ph]}
  \BibitemShut {NoStop}%
\bibitem [{\citenamefont {Vidal}(2003)}]{vidal2003efficient}%
  \BibitemOpen
  \bibfield  {author} {\bibinfo {author} {\bibfnamefont {G.}~\bibnamefont {Vidal}},\ }\href@noop {} {\bibfield  {journal} {\bibinfo  {journal} {Physical review letters}\ }\textbf {\bibinfo {volume} {91}},\ \bibinfo {pages} {147902} (\bibinfo {year} {2003})}\BibitemShut {NoStop}%
\bibitem [{\citenamefont {Ali}(2020)}]{PyCaret}%
  \BibitemOpen
  \bibfield  {author} {\bibinfo {author} {\bibfnamefont {M.}~\bibnamefont {Ali}},\ }\href {https://www.pycaret.org} {\emph {\bibinfo {title} {PyCaret: An open source, low-code machine learning library in Python}}} (\bibinfo {year} {2020}),\ \bibinfo {note} {pyCaret version 1.0}\BibitemShut {NoStop}%
\bibitem [{\citenamefont {Cai}\ \emph {et~al.}(2023)\citenamefont {Cai}, \citenamefont {Babbush}, \citenamefont {Benjamin}, \citenamefont {Endo}, \citenamefont {Huggins}, \citenamefont {Li}, \citenamefont {McClean},\ and\ \citenamefont {O’Brien}}]{cai2023quantum}%
  \BibitemOpen
  \bibfield  {author} {\bibinfo {author} {\bibfnamefont {Z.}~\bibnamefont {Cai}}, \bibinfo {author} {\bibfnamefont {R.}~\bibnamefont {Babbush}}, \bibinfo {author} {\bibfnamefont {S.~C.}\ \bibnamefont {Benjamin}}, \bibinfo {author} {\bibfnamefont {S.}~\bibnamefont {Endo}}, \bibinfo {author} {\bibfnamefont {W.~J.}\ \bibnamefont {Huggins}}, \bibinfo {author} {\bibfnamefont {Y.}~\bibnamefont {Li}}, \bibinfo {author} {\bibfnamefont {J.~R.}\ \bibnamefont {McClean}},\ and\ \bibinfo {author} {\bibfnamefont {T.~E.}\ \bibnamefont {O’Brien}},\ }\href@noop {} {\bibfield  {journal} {\bibinfo  {journal} {Reviews of Modern Physics}\ }\textbf {\bibinfo {volume} {95}},\ \bibinfo {pages} {045005} (\bibinfo {year} {2023})}\BibitemShut {NoStop}%
\bibitem [{\citenamefont {Endo}\ \emph {et~al.}(2021)\citenamefont {Endo}, \citenamefont {Cai}, \citenamefont {Benjamin},\ and\ \citenamefont {Yuan}}]{endo2021hybrid}%
  \BibitemOpen
  \bibfield  {author} {\bibinfo {author} {\bibfnamefont {S.}~\bibnamefont {Endo}}, \bibinfo {author} {\bibfnamefont {Z.}~\bibnamefont {Cai}}, \bibinfo {author} {\bibfnamefont {S.~C.}\ \bibnamefont {Benjamin}},\ and\ \bibinfo {author} {\bibfnamefont {X.}~\bibnamefont {Yuan}},\ }\href@noop {} {\bibfield  {journal} {\bibinfo  {journal} {Journal of the Physical Society of Japan}\ }\textbf {\bibinfo {volume} {90}},\ \bibinfo {pages} {032001} (\bibinfo {year} {2021})}\BibitemShut {NoStop}%
\bibitem [{\citenamefont {Mohri}\ \emph {et~al.}(2018)\citenamefont {Mohri}, \citenamefont {Rostamizadeh},\ and\ \citenamefont {Talwalkar}}]{mohri2018foundations}%
  \BibitemOpen
  \bibfield  {author} {\bibinfo {author} {\bibfnamefont {M.}~\bibnamefont {Mohri}}, \bibinfo {author} {\bibfnamefont {A.}~\bibnamefont {Rostamizadeh}},\ and\ \bibinfo {author} {\bibfnamefont {A.}~\bibnamefont {Talwalkar}},\ }\href {https://mitpress.mit.edu/9780262039406} {\emph {\bibinfo {title} {Foundations of Machine Learning}}},\ \bibinfo {edition} {2nd}\ ed.\ (\bibinfo  {publisher} {The MIT Press},\ \bibinfo {address} {Cambridge, MA},\ \bibinfo {year} {2018})\ \bibinfo {note} {a new edition of a graduate-level machine learning textbook that focuses on the analysis and theory of algorithms.}\BibitemShut {Stop}%
\end{thebibliography}%

\newpage

\onecolumngrid

\appendix
\section{Proof of quantum easiness}\label{sec:proof-quantum-easiness}

To prove Theorem \ref{thm:quantum-easiness}, we first bound the generalization error $R(\bm{w})-\hat{R}(\bm{w})$ through the Rademacher complexity.

\begin{definition}[\cite{mohri2018foundations}, p. 30]
    Let $\mathcal{G}$ be a family of functions mapping from $\mathcal{X}$ to $[a, b]$, and let $S=\left(x_1, \ldots, x_m\right)$ be a fixed sample of size $m$ with elements in $\mathcal{X}$. Then, the empirical Rademacher complexity of $\mathcal{G}$ with respect to the sample $S$ is defined as:
    \begin{align}
        \widehat{\mathfrak{R}}_S(\mathcal{G})=\underset{\boldsymbol{\sigma}}{\mathbb{E}}\left[\sup _{g \in \mathcal{G}} \frac{1}{m} \sum_{i=1}^m \sigma_i g\left(x_i\right)\right]
    \end{align}
    with $\sigma_i$'s being independent uniform random variables taking values in $\{-1,+1\}$.
\end{definition}

The Rademacher complexity can be used to bound the generalization error of a model by the following lemma.

\begin{lemma}[\cite{mohri2018foundations}, p. 270]\label{lemma:generalization-bound}
        Let $ \mathcal{H} $ be a family of functions that maps the input space $\mathcal{X}$ to $\mathcal{Y}\subset\mathbb{R}$, and assume that for all $ (x, y) \in \mathcal{X} \times \mathcal{Y} $ and $ h \in \mathcal{H} $, $ |h(x) - y| \leq M $ holds for some $ M > 0 $. Let $ p \geq 1 $. Then, with probability at least $ 1 - \delta $ over a sample $ S \in \mathcal{X} \times \mathcal{Y}$ of size $ m $ drawn from the distribution $\mathcal{D}$ on $\mathcal{X} \times \mathcal{Y}$, the following inequality holds for all $ h \in \mathcal{H} $:
        \[
        \mathbb{E}_{(x, y) \sim \mathcal{D}} \left[ |h(x) - y|^p \right] \leq \frac{1}{m} \sum_{i=1}^m |h(x_i) - y_i|^p + 2p M^{p-1} \widehat{\mathfrak{R}}_S(\mathcal{H}) + 3M^p \sqrt{\frac{\log \frac{2}{\delta}}{2m}},
        \]
        where $\widehat{\mathfrak{R}}_m(\mathcal{H})$ denotes the empirical Rademacher complexity of the hypothesis class $\mathcal{H}$.
\end{lemma}

This theorem directly implies $R(\bm{w})-\hat{R}(\bm{w}) \leq 4M\widehat{\mathfrak{R}}_S(\mathcal{H}) + 3M^2\sqrt{\frac{\log \frac{2}{\delta}}{2N_d}}$ by setting $p=2$.
It thus remains to bound the empirical Rademacher complexity $\widehat{\mathfrak{R}}_S(\mathcal{H})$ and $|h(x) - y|$ for our specific model $g(H,\rho;\bm{w})$ for making a bound on $R(\bm{w})-\hat{R}(\bm{w})$.
We use the following theorem from \cite{mohri2018foundations} to do the former.

\begin{lemma}[\cite{mohri2018foundations}, p. 118]
    Let $K: \mathcal{X} \times \mathcal{X} \rightarrow \mathbb{R}$ be a positive definite symmetric kernel and let $\mathbb{H}$ be a reproducing kernel Hilbert space associated to $K$ and $\Phi: \mathcal{X} \rightarrow \mathbb{H}$ be a feature mapping associated to $K$. Let $S$ be a subset of $\mathcal{X}$ with size $m$, and let $\mathcal{H}=\left\{x \mapsto\langle\bm{a}, \Phi(x)\rangle_{\mathbb{H}}:\|\bm{a}\|_{\mathbb{H}} \leq \Lambda\right\}$ for some $\Lambda \geq 0$, where $\braket{\cdot,\cdot}_{\mathbb{H}}$ and $\|\cdot\|_{\mathbb{H}}$ denote the inner product and the norm defined in $\mathbb{H}$. Then,
    \begin{align}
        \widehat{\mathfrak{R}}_S(\mathcal{H}) \leq \frac{\Lambda \sqrt{\operatorname{Tr}[\mathbf{K}]}}{m},
    \end{align}
    where $\mathbf{K}$ is the kernel matrix with elements $K(x_i, x_j)$ for $x_i,x_j \in S$.
\end{lemma}
Note that it immediately implies the following corollary.
\begin{corollary}\label{cor:rademacher}
    Let $\Phi: \mathcal{X} \rightarrow \mathbb{R}^n$ be a feature map. Let $S$ be a subset of $\mathcal{X}$ with size $m$, and let $\mathcal{H}=\left\{x \mapsto \bm{a}^{\T}\Phi(x): \bm{a}^{\T} \bm{a} \leq \Lambda^2\right\}$ for some $\Lambda \geq 0$. Let $\mathbf{K}$ be the kernel matrix with elements $K_{ij} = \Phi(x_i)^{\T}\Phi(x_j)$ for $x_i,x_j \in S$. Then,
    \begin{align}
        \widehat{\mathfrak{R}}_S(\mathcal{H}) \leq \frac{\Lambda \sqrt{\operatorname{Tr}[\mathbf{K}]}}{m},
    \end{align}
\end{corollary}

To bound the Rademacher complexity of our model, we thus want to bound $\operatorname{Tr}[\mathbf{K}]$ and use Corollary \ref{cor:rademacher}.
In our case, the kernel matrix elements can be calculated as
\begin{align}
    K_{ij} 
    &= \sum_{k=0}^{2K}x_k(H_i,\rho_i)x_k(H_j,\rho_j),
\end{align}
where $x_k$ is defined as \eqref{eq:fourier-feature}.
Since $\|\cos(A)\|\leq 1$ and $\|\sin(A)\|\leq 1$ for any Hermitian operator $A$, $\Tr[\cos(k\pi H/C)\rho] \leq \Tr[\rho] = 1$ and $\Tr[\sin(k\pi H/C)\rho] \leq \Tr[\rho] = 1$ hold for any state $\rho$.
Thus, $|x_k(H,\rho)|\leq 1$ for any Hamiltonian $H$, state $\rho$, and integer $k$. Consequently, $K_{ij} \leq 2K+1$.
Furthermore, since in our case $\mathbf{K}$ is $N_d\times N_d$ matrix, we have $\operatorname{Tr}[\mathbf{K}] \leq (2K+1)N_d$, which in turn implies $\widehat{\mathfrak{R}}_S(\mathcal{H})\leq W\sqrt{\frac{2K+1}{N_d}}$ with the assumption $\bm{w}^\T \bm{w}\leq W^2$.

We then bound $|h(x) - y|$.
First, $\mathcal{Y}$ in our setting is defined by $\mathrm{Tr}[f(H)\rho]$ which can be bounded as $\mathrm{Tr}[f(H)\rho]\leq \|f(H)\| \leq \|f\|_\infty$.
This implies $y\leq \|f\|_\infty$.
Second, when $\bm{w}^\T\bm{w}\leq W^2$, the model $g(H,\rho;\bm{w})=\sum_{k=0}^{2K} w_k x_k(H,\rho)$ satisfies $|g(H,\rho;\bm{w})|\leq W\sqrt{2K+1} $ from the well-known inequality $\|\bm{w}\|_1 \leq \sqrt{2K+1} \|\bm{w}\|_2$ and that $|\mathrm{Tr}[\cos(k\pi H/C)\rho]|\leq 1$ and $|\mathrm{Tr}[\sin(k\pi H/C)\rho]|\leq 1$.
These two allow us to establish a bound
\begin{align} \label{eq:model-output-bound}
    |h(x) - y| \leq W\sqrt{2K+1}+\|f\|_\infty
\end{align}
by the triangle inequality.

Let us now summarize the discussion as the following lemma, which can be obtained by combining the bound on $\widehat{\mathfrak{R}}_S(\mathcal{H})$ and $|h(x) - y|$ with Lemma \ref{lemma:generalization-bound}:
\begin{lemma}\label{lemma:generalization-bound-this-work}
    Let $g(H,\rho;\bm{w})=\sum_{k=0}^{2K} w_k x_k(H,\rho)$ and $L(y, y')=|y-y'|^2$ be the loss function.
    Let $p(H,\rho)$ be the probability distribution of Hamiltonians and states satisfying $\|H\|\leq C$.
    Define the expected loss as $R(\bm{w}) = \mathbb{E}_{H,\rho\sim p(H,\rho)}\left[L(\mathrm{Tr}[f(H)\rho], g(H,\rho;\bm{w}))\right]$.
    For samples $\{(H_i,\rho_i)\}_{i=1}^{N_d}$ drawn independently from $p(H,\rho)$, with probability at least $1-\delta$, we have
    \begin{align}\label{eq:generalization-error-bound}
        R(\bm{w}) - \hat{R}(\bm{w}) \leq 4(W\sqrt{2K+1} +\|f\|_\infty)W\sqrt{\frac{2K+1}{N_d}} + 3(W\sqrt{2K+1} +\|f\|_\infty)^2\sqrt{\frac{\log \frac{2}{\delta}}{2N_d}},
    \end{align}
    where $\hat{R}(\bm{w}) = \frac{1}{N_d}\sum_{j=1}^{N_d} L(y_j, g(H_j,\rho_j;\bm{w}))$ is the empirical loss, for all $\bm{w}\in \mathbb{R}^{2K+1}$ such that $\bm{w}^\T\bm{w}\leq W^2$.
\end{lemma}

We now show that the empirical loss can be made small using the model $g(H,\rho;\bm{w})=\sum_{k=0}^K w_k x_k(H,\rho;\bm{w})$.
Let $\lambda_l$ be eigenvalues of Hamiltonian $H$ and $\ket{\lambda_l}$ be corresponding eigenstates.
Let $p_l = \braket{\lambda_l|\rho|\lambda_l}$.
Then, $\mathrm{Tr}[f(H)\rho] = \sum_{l=0}^{2^n-1} f(\lambda_l)p_l$ and  $g(H,\rho;\bm{w})=\sum_{l=0}^{2^n-1} \left[\sum_{k=0}^K w_{2k}\cos(k\pi \lambda_l/C) + \sum_{k=0}^{K-1} w_{2k+1}\sin(k\pi \lambda_l/C)\right]p_l$.
The empirical loss we minimize can then be written as
\begin{align}
    \hat{R}(\bm{w}) 
    &= \frac{1}{N_d}\sum_{i=1}^{N_d} \left(\sum_{l=0}^{2^n-1} f(\lambda_{i,l})p_l - \sum_{l=0}^{2^n-1} \left[\sum_{k=0}^K w_{2k}\cos(k\pi \lambda_{i,l}/C) + \sum_{k=0}^{K-1} w_{2k+1}\sin(k\pi \lambda_{i,l}/C)\right]p_l\right)^2 \\
    &= \frac{1}{N_d}\sum_{i=1}^{N_d} \left[\sum_{l=0}^{2^n-1} p_l\left(f(\lambda_{i,l}) - \left[\sum_{k=0}^K w_{2k}\cos(k\pi \lambda_{i,l}/C) + \sum_{k=0}^{K-1} w_{2k+1}\sin(k\pi \lambda_{i,l}/C)\right]\right)\right]^2 \label{eq:empirical-loss-decomposed}
\end{align}
From the assumption of the theorem, $\hat{R}(\bm{w}^*) \leq \hat{R}(\bm{w})$ for all $\bm{w}$ such that $\bm{w}^\T\bm{w}\leq W^2$, because $\bm{w}^*$ is a minimizer of $\hat{R}(\bm{w}^*)$.
Therefore, $\hat{R}(\bm{w}^*) \leq \hat{R}(\bm{c}_K)$.
Moreover, $\hat{R}(\bm{c}_K)\leq\varepsilon_K^2$ from the assumption of the theorem that $\left|f(x) - \left[\sum_{k=0}^K w_{2k}\cos(k\pi x/C) + \sum_{k=0}^{K-1} w_{2k+1}\sin(k\pi x/C)\right]\right|\leq \varepsilon_K$.
Thus, $\hat{R}(\bm{w}^*) \leq \varepsilon_K^2$.
Combining this with the generalization error bound \eqref{eq:generalization-error-bound}, we have proved Theorem~\ref{thm:quantum-easiness}.

\section{Learnability when features are estimated with statistical noise}
\label{sec:proof-statistical-easy}

Here, we prove the easiness of the ML task when the statistical noise affects the elements of Hamiltonian Fourier features $\bm{x}$.
First, note that we can efficiently obtain an estimate $\tilde{\bm{x}}$ of the feature vector $\bm{x}$ such that each element satisfies $|\tilde{x}_k-x_k|\leq \eta$ for some small $\eta>0$ with $\mathcal{O}(1/\eta^2)$ shots with high probability as we will discuss rigorously later.
This means that the remaining challenge in proving easiness is to show that the expected loss can still be made small when each element \( x_k \) of the feature vector \( \bm{x} \) is perturbed by an error of at most \( \eta \).

To prove this, we first analyze how the empirical loss changes when clean features $\bm{x}$ are replaced by their noisy estimates $\tilde{\bm{x}}$, and we show that the resulting difference in loss is bounded in terms of the noise level $\eta$. 
This is done by deriving a deterministic bound on the deviation between the empirical losses with and without noise. 
Then, we extend this result to the expected loss using standard generalization bounds. 
Finally, we combine this with an analysis of the number of quantum measurements, or shots, required to ensure that the estimation error $\epsilon_k$ for each feature remains within $\eta$ with high probability. 
This leads to a complete bound on the expected loss achieved by minimizing the empirical loss over the noisy features, thereby demonstrating learnability under statistical noise.

\begin{lemma}
\label{lemma:stat-error-1}
Let the notations follow those of Theorem \ref{thm:quantum-easiness}.
Furthermore, let $\tilde{\bm x}=\bm x+\bm\epsilon$ be noisy feature vector, where $\epsilon_k$, an element of $\bm\epsilon$, satisfies $|\epsilon_k|\leq \eta$.
Let $\tilde R(\bm w):=\frac{1}{N_d} \sum_{i=1}^{N_d}\left(y_i-\bm w^{\T}\tilde{\bm x}_i\right)^2$.
Then,
\begin{align}
    |\hat R(\bm w)-\tilde R(\bm w)| \leq 2\eta W\sqrt{2K+1} \left(\|f\|_\infty+W\sqrt{2K+1}\right)
    +\eta^{2}W^{2}(2K+1)
\end{align}
holds for any $\bm{w}$.
\end{lemma}

\begin{proof}
For any $\bm{w}$,
\begin{align}
  \hat R(\bm w)-\tilde R(\bm w)
  &=\frac{1}{N_d}\sum_{i=1}^{N_d}[(y_i-\bm{w}^\T\bm{x}_i)^2 - (y_i-\bm{w}^\T\tilde{\bm{x}_i})^2]. \\
  &=\frac{1}{N_d}\sum_{i=1}^{N_d}[(y_i-\bm{w}^\T\bm{x}_i)^2 - ((y_i-\bm{w}^\T\bm{x}_i)-\bm{w}^\T\bm{\epsilon}_i)^2]. \\
  &= \frac{1}{N_d}\sum_{i=1}^{N_d}[2(y_i-\bm{w}^\T\bm{x}_i)\bm{w}^\T\bm{\epsilon}_i-(\bm{w}^\T\bm{\epsilon}_i)^2].\label{eq:stat-error-lamma-1}
\end{align}
Because \(|y_i|\le \|f\|_{\infty}\) and \(|x_{ik}|\le1\),
\begin{align}\label{eq:stat-error-lamma-2}
    \bigl|y_i-\bm w^{\mathsf T}\bm x_i\bigr|
  \;\le\;
  \|f\|_\infty +\|\bm w\|_2\|\bm x_i\|_2
  \;\le\;
  \|f\|_\infty + W\sqrt{2K+1}.
\end{align}
Also, from the assumption on $\bm{\epsilon}$,
\begin{align}\label{eq:stat-error-lamma-3}
      |\bm{w}^\T\bm{\epsilon}_i| \leq 
  \|\bm w\|_2\|\bm\varepsilon_i\|_2
  \leq
  \eta W\sqrt{2K+1}.
\end{align}
Eqs. \eqref{eq:stat-error-lamma-1}, \eqref{eq:stat-error-lamma-2} and \eqref{eq:stat-error-lamma-3} implies the theorem.

\end{proof}

This lemma immedeately implies the following:

\begin{lemma}\label{lemma:stat-error-2}
Let the notations follow those of \ref{lemma:stat-error-1}.
With probability at least \(1-\delta\), the following holds:
\begin{align}
    \begin{split}
        |R(\bm{w}) - \tilde{R}(\bm{w})| &\leq 4(W\sqrt{2K+1} +\|f\|_\infty)W\sqrt{\frac{2K+1}{N_d}} + 3(W\sqrt{2K+1} +\|f\|_\infty)^2\sqrt{\frac{\log \frac{2}{\delta}}{2N_d}} \\
    &\quad \quad + 2\eta W\sqrt{2K+1} \left(\|f\|_\infty+W\sqrt{2K+1}\right)
    +\eta^{2}W^{2}(2K+1),
    \end{split}
\end{align}
\end{lemma}
\begin{proof}
First, write
\begin{align}
R(\bm{w})-\tilde R(\bm{w})
  =[R(\bm{w})-\hat R(\bm{w})]
  +[\hat R(\bm{w})-\tilde R(\bm{w})].
\end{align}
We use Lemma \ref{lemma:generalization-bound-this-work} for the first term and Lemma \ref{lemma:stat-error-1} for the second term to bound $|R(\bm{w})-\tilde R(\bm{w})|$.
\end{proof}

Using the previous results, we now show that the expected loss $R(\tilde{\bm{w}}^*)$ of the weight vector $\tilde{\bm{w}}^*$ trained on noisy features remains small. In other words, we formally demonstrate that learnability is preserved even in the presence of statistical noise in the features, as stated in the following lemma.

\begin{lemma}
    Let the notations follow those of \ref{lemma:stat-error-1}.
    Let $\tilde{\bm{w}}^*$ be the minimizer of $\tilde{R}(\bm{w})$. With probability at least $1-\delta$, the following holds:
    \begin{align}
        \begin{split}
        R(\tilde{\bm{w}}^*) &\leq \varepsilon_K^2 + 4(W\sqrt{2K+1} +\|f\|_\infty)W\sqrt{\frac{2K+1}{N_d}} + 3(W\sqrt{2K+1} +\|f\|_\infty)^2\sqrt{\frac{\log \frac{2}{\delta}}{2N_d}} \\
        &\quad \quad + 4\eta W\sqrt{2K+1} \left(\|f\|_\infty+W\sqrt{2K+1}\right)
        +2\eta^{2}W^{2}(2K+1).
        \end{split}
    \end{align}
\end{lemma}
\begin{proof} 
    By Lemma \ref{lemma:stat-error-2}, 
    \begin{align}\label{eq:stat-error-easiness-1}
    \begin{split}
        R(\tilde{\bm{w}}^*) &\leq \tilde{R}(\tilde{\bm{w}}^*) + 4(W\sqrt{2K+1} +\|f\|_\infty)W\sqrt{\frac{2K+1}{N_d}} + 3(W\sqrt{2K+1} +\|f\|_\infty)^2\sqrt{\frac{\log \frac{2}{\delta}}{2N_d}} \\
    &\quad \quad + 2\eta W\sqrt{2K+1} \left(\|f\|_\infty+W\sqrt{2K+1}\right)
    +\eta^{2}W^{2}(2K+1).
    \end{split}
    \end{align}
    Now, note that 
    \begin{align}\label{eq:stat-error-easiness-2}
        \tilde{R}(\tilde{\bm{w}}^*)\leq\tilde{R}(\bm{w}^*)
    \end{align} 
    by defintion of $\tilde{\bm{w}}^*$.
    By Lemma \ref{lemma:stat-error-1}, we have 
    \begin{align}\label{eq:stat-error-easiness-3}
        \tilde{R}(\bm{w}^*) \leq \hat{R}(\bm{w}^*) + 2\eta W\sqrt{2K+1} \left(\|f\|_\infty+W\sqrt{2K+1}\right)+\eta^{2}W^{2}(2K+1).
    \end{align}
    From the assumption of the theorem, $\hat{R}(\bm{w}^*) \leq \hat{R}(\bm{w})$ for all $\bm{w}$ such that $\bm{w}^\T\bm{w}\leq W^2$, because $\bm{w}^*$ is a minimizer of $\hat{R}(\bm{w}^*)$.
    Therefore, $\hat{R}(\bm{w}^*) \leq \hat{R}(\bm{c}_K)$.
    Moreover, $\hat{R}(\bm{c}_K)\leq\varepsilon_K^2$ from the assumption of the theorem that $\left|f(x) - \left[\sum_{k=0}^K w_{2k}\cos(k\pi x/C) + \sum_{k=0}^{K-1} w_{2k+1}\sin(k\pi x/C)\right]\right|\leq \varepsilon_K$.
    Thus, 
    \begin{align}\label{eq:stat-error-easiness-4}
        \hat{R}(\bm{w}^*) \leq \varepsilon_K^2.
    \end{align}
    Combining \eqref{eq:stat-error-easiness-1}, \eqref{eq:stat-error-easiness-2}, \eqref{eq:stat-error-easiness-3} and \eqref{eq:stat-error-easiness-4}, we obtain the lemma.
\end{proof}

We now discuss the number of shots needed to suppress the statistical error $\epsilon_k$ to $\epsilon_k\leq \eta$ for all $k$ with high probability, say $1-\delta$.
To do so, we must fix the technique we use to estimate $x_k$.
Here, we assume the use of Hadamard test and consequently the following strategy to estimate $x_k$.
To estimate $x_{\cos, k}$, for a given Hamiltonian $H$ and a state $\rho$, we execute the circuit in Fig. \ref{fig:hadamard-test} repeatedly for $N_{\rm shot}$ times with $t=k\pi/C$.
Define a random variable $z_m$ for $m=1,...,N_{\rm shot}$ that takes $z_m=+1$ if we obtain $\ket{+}$ and $z_m=-1$ if we obtain $\ket{-}$ at the $m$th trial.
Then, we construct an estimator 
\begin{align}\label{eq:hadamard-test-estimator}
    \tilde{x}_{\cos, k} := \frac{1}{N_{\rm shot}}\sum_{m=1}^{N_{\rm shot}} z_m.
\end{align}
It is well-known that $\tilde{x}_{\cos, k}$ is an unbiased estimator of $x_{\cos, k}$.
Hoeffding's inequality tells us that 
\begin{align}
    \epsilon_k = |\tilde{x}_{\cos, k} - x_{\cos, k}| \leq \eta
\end{align}
holds with probability at least $1-2\exp(-N_{\rm shot}\eta^2/2)$.
We can take the same strategy to construct an estimator $\tilde{x}_{\sin, k}$ for $x_{\sin, k}$ by changing the measurement basis to $\ket{\pm i}$ in Fig. \ref{fig:hadamard-test}.
Note that the order of statistical error remains same for the strategy taken in the main text, since $w_{\pm}$ and $w_{\pm i}$ are estimated as the probability of getting the corresponding states $\ket{\psi_{\pm}}$ and $\ket{\psi_{\pm i}}$ from $e^{-iHt}\ket{\psi_+}$ by sampling.

Now, we shall take $N_{\rm shot}$ for each $k$ to ensure a single $\tilde{x}_k$ have error at most $\eta$ with probability $1-\delta/(2K+1)$, which ensure all $2K+1$ estimators $\tilde{x}_k$ have error at most $\eta$ with probability $1-\delta$ by union bound.
$N_{\rm shot}$ needed to guarantee the above is $N_{\rm shot}=\mathcal{O}(\log(K/\delta)/\eta^2)$.
This discussion implies the following theorem.

\begin{theorem}
    Let the notations follow those of Theorem \ref{thm:quantum-easiness}.
    Let $\delta>0$ and $\eta>0$.
    Let $N_{\rm shot}=\mathcal{O}(\log(K/\delta)/\eta^2)$.
    Let $\tilde{\bm{x}}_i$ be an estimator of $\bm{x}_i$ obtained via performing $N_{\rm shot}$-shot Hadamard test as in Eq. \eqref{eq:hadamard-test-estimator}.
    Let $\tilde R(\bm w):=\frac{1}{N_d} \sum_{i=1}^{N_d}\left(y_i-\bm w^{\T}\tilde{\bm x}_i\right)^2$ and $\tilde{\bm{w}}^*$ be minimizer of $\tilde R(\bm w)$.
    Then, with probability $1-2\delta$, the following holds:
    \begin{align}
        \begin{split}
        R(\tilde{\bm{w}}^*) &\leq \varepsilon_K^2 + 4(W\sqrt{2K+1} +\|f\|_\infty)W\sqrt{\frac{2K+1}{N_d}} + 3(W\sqrt{2K+1} +\|f\|_\infty)^2\sqrt{\frac{\log \frac{2}{\delta}}{2N_d}} \\
        &\quad \quad + 4\eta W\sqrt{2K+1} \left(\|f\|_\infty+W\sqrt{2K+1}\right)
        +2\eta^{2}W^{2}(2K+1).
        \end{split}
    \end{align}
\end{theorem}

\begin{figure}
    \centering
    \includegraphics{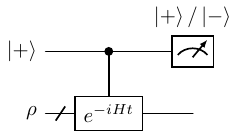}
    \caption{Quantum circuit of Hadamard test to estimate $x_{\cos, k}$. We first prepare the top qubit and $n$-qubit register in $\ket{+}$ and $\rho$ states, respectively. Then, after applying controlled time evolution, we measure the top qubit in $\ket{\pm}$ basis.}
    \label{fig:hadamard-test}
\end{figure}

\section{Details of experiments}\label{sec:demo-setup}

The status of the 40 qubits used in this experiment on \textit{ibmq\_marrakesh} was as follows: a median readout error of $8.79\times10^{-3}$, $\text{T}_1$ of 219.4 $\mathrm{\mu}$s, $\text{T}_2$ of 130.1 $\mathrm{\mu}$s for the qubits utilized in this study, and a median CZ gate error of $2.72\times10^{-3}$ for the employed qubit pairs.
We used logical-to-physical mapping of qubits as shown in Fig.~\ref{fig:layout}.
We transpiled logical circuits to executable circuits using Qiskit \cite{qiskit2024}.
To compute $x_{\mathrm{cos}, l}$ and $x_{\mathrm{sin}, l}$ for $l=0,...,11$, we heuristically set the number of Trotter steps, $n_{\rm step}$, to [1, 1, 1, 1, 1, 2, 2, 2, 2, 3, 3, 3] for 12 qubits, to [1, 1, 2, 2, 2, 2, 3, 3, 3, 3, 4, 4] for 32 qubits, and to [1, 1, 2, 2, 2, 3, 3, 4, 4, 4, 5, 5] for 40 qubits.
An example of the transpilation is shown in Fig.~\ref{fig:example_circuits}.

\begin{figure}
    \centering
    \includegraphics[width=1.0\linewidth]{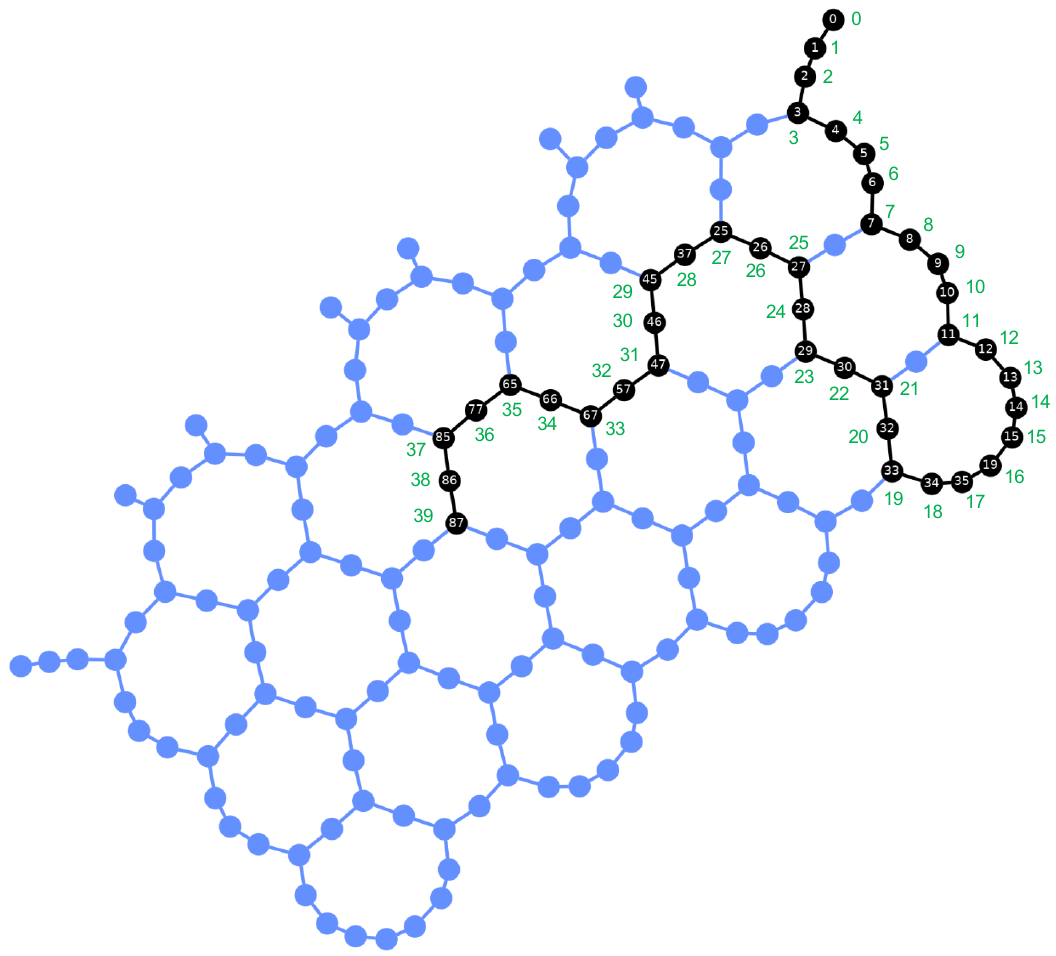}
    \caption{Mapping from logical qubits (green numbers) to physical qubits (white numbers) on \textit{ibmq\_marrakesh} for a 40-qubit configuration. The same mapping is used for the 12- and 32-qubit configurations.}
    \label{fig:layout}
\end{figure}

\begin{figure}
\centering
\subfloat[logical circuit\label{subfig:example_logical_circuit}]{
    \includegraphics[width=\textwidth]{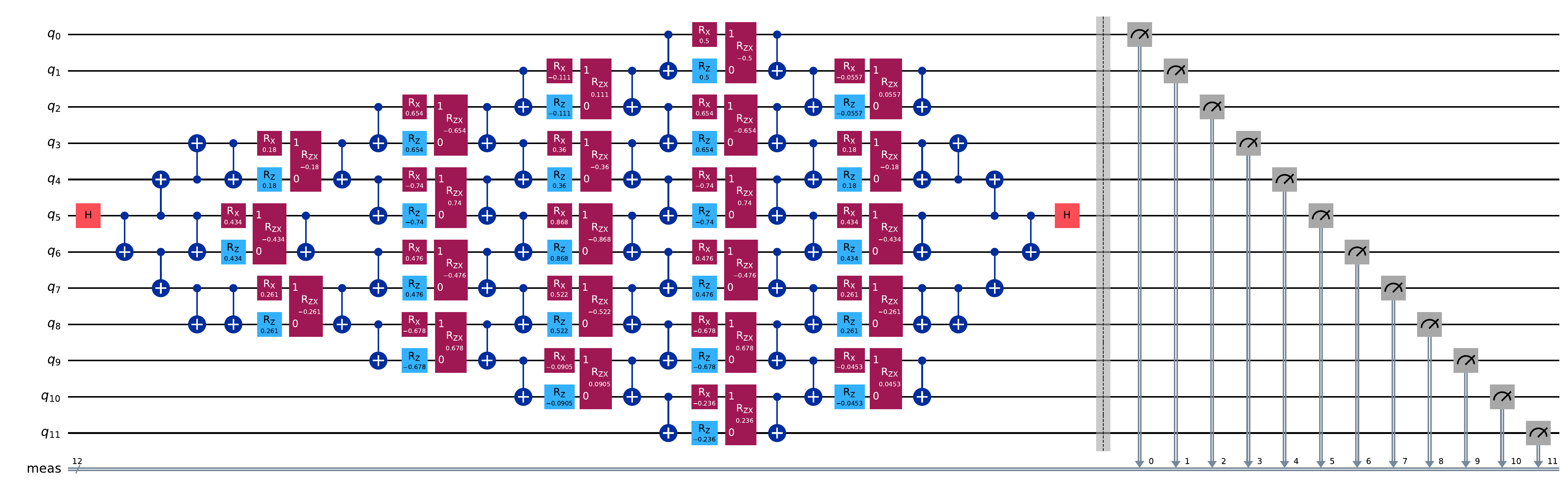}
}\\
\subfloat[physical circuit\label{subfig:example_physical_circuit}]{
    \includegraphics[width=\textwidth]{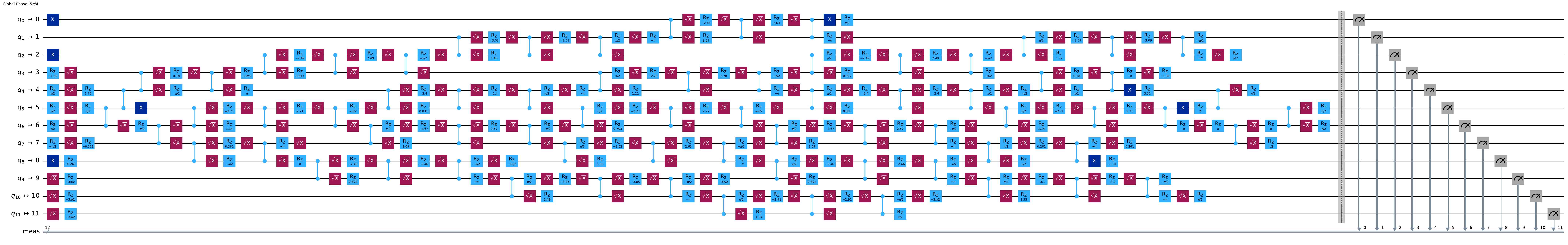}
}

\caption{
Examples of a logical circuit (before transpilation) and the corresponding physical circuit (after transpilation) used in this study for 12 qubits. The depicted quantum circuits compute the $l=5$th feature, corresponding to a Trotter step value of $n_{\text{step}}=2$.
}
\label{fig:example_circuits}
\end{figure}

The regression methods compared in this study are those provided by PyCaret, including Bayesian Ridge, Huber, Linear Regression, Gradient Boosting, Extra Trees, Ridge, Passive Aggressive, Decision Tree, Random Forest, AdaBoost, K Neighbors, Orthogonal Matching Pursuit, Light Gradient Boosting Machine, Lasso, Elastic Net, Lasso Least Angle Regression, and Least Angle Regression. 
These methods are evaluated using k-fold cross-validation with the number of folds set to the default value of 10. All other PyCaret settings are kept at their default values.

\end{document}